\documentclass[10pt,aps,pra,twocolumn,floatfix,nofootinbib,superscriptaddress,longbibliography]{revtex4-2}
\usepackage{placeins}

\usepackage{bm,graphicx,mathrsfs,amsmath,amssymb,mathtools,,makecell,bbm, amsthm, bm,dsfont,color,times,txfonts,nicefrac,framed,enumitem,tikz,physics,wrapfig,amsfonts,tcolorbox,lipsum,nicematrix}
\usepackage[ruled,vlined, linesnumbered]{algorithm2e}

\usepackage[colorlinks=true,linkcolor=equationcolor,citecolor=teal]{hyperref}

\hypersetup{urlcolor=teal}
\definecolor{equationcolor}{RGB}{222,94,100}
\definecolor{alecolor}{RGB}{198,113,190}
\definecolor{changescolor}{rgb}{0, 0, 0.7}

\makeatletter
\def\blfootnote{\gdef\@thefnmark{}\@footnotetext}
\makeatother

\newcommand{\ms}[1]{\textsf{#1}}

\newcommand{\E}[1]{\mathcal{E}}

\def\E{ {\ms E} }

\newtheorem{thm}{Theorem}

\newtheorem{lem}{Lemma}
\newtheorem{prop}{Proposition}
\newtheorem{cor}{Corollary}

\newtheorem{defn}{Definition}

\begin{document}
\title{Entanglement manipulation and magic area-laws}

\author{Rafael A. Macêdo}
\affiliation{Physics Department, Federal University of Rio Grande do Norte, Natal, 59072-970, Rio Grande do Norte, Brazil}
\affiliation{International Institute of Physics, Federal University of Rio Grande do Norte, 59078-970, Natal, Brazil}
\affiliation{Department of Mathematical Sciences, University of Copenhagen, Universitetsparken 5, 2100 Copenhagen,
Denmark}
\author{Rafael Chaves}
\affiliation{International Institute of Physics, Federal University of Rio Grande do Norte, 59078-970, Natal, Brazil}

\begin{abstract}
Physical many-body quantum systems typically interact locally, giving rise to ground states where entanglement grows with the area of the region being probed. Beyond entanglement, non-stabilizerness (or "magic") has been established as another crucial resource linked to quantum complexity and the properties of error-correcting states. Unlike entanglement, non-stabilizerness typically behaves with a volume-law scaling in ground states, indicating inherent computational complexity. In this work, we show that under a recent classification of quantum states, area-law states split into two distinct classes. Entanglement-dominated states obey an area-law for non-stabilizerness and admit efficient quantum algorithms for entanglement manipulation. Conversely, magic-dominated states exhibit a a faster-than-area growth of subsystem non-stabilizerness and require non-efficient algorithms for entanglement manipulation tasks. Within this context, we introduce the concept of the stabilizer regime, a restricted quantum phase that admits efficient entanglement manipulation and hosts area laws for both magic and entanglement.
\end{abstract}

\maketitle
\section{Introduction}

Quantum many-body phenomena naturally arise in states that inhabit Hilbert spaces whose dimensions grow exponentially with the number of constituents. However, most of these states are unstructured. For instance, when chosen at random, they exhibit extensive entanglement, known as a volume law, and exponentially small expectation values for local observables \cite{page1993average, foong1994proof, hayden2006aspects, mele2024introduction}. In contrast, physical quantum systems are typically governed by Hamiltonians with local interactions defined by some underlying geometry. The ground states of such systems are known to obey an entanglement area-law \cite{eisert2010colloquium}, in which the entanglement of a subregion scales only with the size of its boundary. This behavior implies that these states are highly compressible and occupy a structured corner of the Hilbert space.

The structured entanglement of ground states arising from local interactions is a foundational feature in the classification of quantum phases of matter \cite{zeng2019quantum}, which aims to categorize ground states into equivalence classes connected by adiabatic evolution. This has led to a computationally inspired definition of a quantum phase \cite{chen2010local}, in which two states belong to the same phase if they can be transformed into each other by a finite-depth quantum circuit (FDQC). Such transformations preserve local correlations and, crucially, maintain the area-law.

This computational definition frames quantum phases in terms of quantum complexity \cite{baiguera2025quantum}, placing states that are related by a short-depth unitary into the same phase. However, entanglement is not the only source of complexity in quantum states. The Gottesman-Knill theorem \cite{gottesman1998heisenberg, aaronson2004improved} famously shows that quantum circuits composed entirely of Clifford group operations, can be efficiently simulated on classical computers, despite these states often being highly entangled. Stabilizer states also occupy a central role in quantum computation, forming the backbone of quantum error-correcting codes \cite{nielsen2010quantum}, which are essential for the reliable storage and manipulation of quantum information. In light of these dual roles, both as classically simulable states and as components of robust quantum error correction, a resource theory of non-stabilizerness, colloquially known as “magic”, has been developed \cite{veitch2014resource, howard2017application}. Non-stabilizerness has since emerged as a key quantifier of the non-classical resources required for universal quantum computation. It is the ingredient that enables quantum states to go beyond classical simulability and the resource that must be injected into stabilizer codes to achieve fault-tolerant quantum computation \cite{bravyi2005universal}.

Among the various proposed measures, the stabilizer Rényi entropy
\cite{leone2022stabilizer, leone2024stabilizer} is of particular relevance, serving as a faithful measure of non-stabilizerness analogous to the entanglement Rényi entropies. It can be efficiently computed classically for matrix product states \cite{haug2023quantifying,tarabunga2024nonstabilizerness, frau2025stabilizer}, leading to a body of numerical work estimating its behaviour in one-dimensional quantum systems \cite{oliviero2022magic, frau2024nonstabilizerness}. Overall, the stabilizer Rényi entropy is typically extensive, coinciding with the generic behaviour for random states \cite{leone2022stabilizer}. This suggests that even small perturbations to a local stabilizer Hamiltonian generate an extensive amount of magic in the ground state, while the system remains in the same phase.

This observation implies that ground state entanglement is robust whereas its magic is not. This is of practical relevance, since robust stabilizer error-correcting codes are also representatives of quantum phases \cite{zeng2019quantum}, such as $\mathbb Z_2$ topological order represented by the toric code \cite{kitaev2003fault}. Consequently, while global properties of those codes are preserved throughout the phase, their stabilizerness is lost upon adding perturbations, complicating information retrieval since standard decoders heavily rely on the stabilizer formalism. It is therefore natural to ask under what conditions area-law states mantain their stabilizerness protected. 

Recent advances in understanding the entanglement structure of non-stabilizer states have begun to shed light on this question. In Ref.~\cite{gu2024magic}, a classification was proposed that distinguishes between two types of states: entanglement-dominated (ED) states, where entanglement entropy outweighs the contribution from magic, and magic-dominated (MD) states, where the opposite holds. While efficient algorithms for entanglement extraction, dilution, and estimation exist in the ED regime, these tasks become provably computationally hard in the MD regime. Physical examples of MD states include Haar-random states \cite{mele2024introduction}, which represent chaotic, long-time many-body quantum systems, whereas ED states include near-stabilizer codes where entanglement stems predominantly from their stabilizer structure.

In this work, we explore how the ED/MD decomposition of Hilbert space manifests under the constraint of an entanglement area-law. Our main results, illustrated in Fig.~\ref{fig:theorem}, show that area-law quantum states inherit the ED/MD dichotomy. Specifically, states in the ED phase exhibit an area-law for the stabilizer Rényi entropy, whereas states in the MD phase show a stabilizer entropy that grows faster than the boundary area. Furthermore, we show that even under the area-law constraints, entanglement manipulation in MD states remains hard. This finding implies that quantum ground states, even those with stabilizer fixed points, typically fall into the MD category: only specific classes of perturbations can drive stabilizer states into the ED regime, where their structure can be effectively harnessed for quantum information processing.

\begin{figure*}
    \centering
    \includegraphics[width=0.6\linewidth]{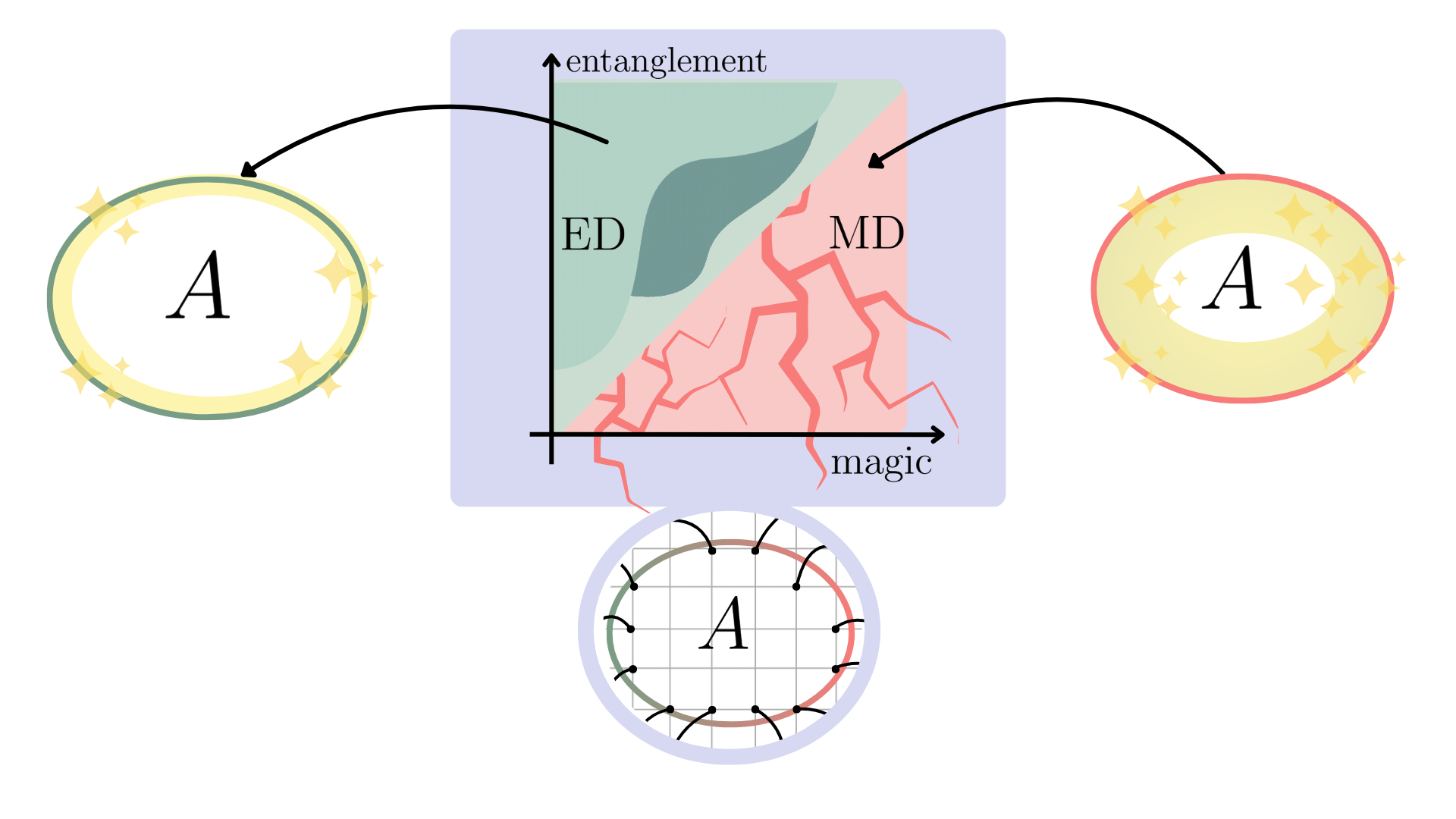}
    \caption{
    Area-law states, illustrated at the bottom as quantum states with entanglement bounded by the perimeter of a subsystem $A$ (with dangling bonds corresponding to fictitious Bell pairs) , can be classified into entanglement-dominated (ED) and magic-dominated (MD) phases, depending on the relative magnitude of their bipartite entanglement and non-stabilizerness (magic). On the left, we show that all ED area-law states exhibit an area-law for the stabilizer Rényi entropy. On the right, we illustrate that any state with a strictly faster-than-area growth of stabilizer entropy must belong to the MD phase. This dichotomy motivates a refined notion of quantum phases, where ED states share complexity properties—particularly for information-processing tasks—with stabilizer states. The color coding follows the convention introduced in Ref.~\cite{gu2024magic}, with green denoting ED states and red denoting MD states.}
    \label{fig:theorem}
\end{figure*}

This work is organized as follows.  In Sec. \ref{sec:phasesandscaling}, we review the necessary background, introduce our notation, and formalize the key definitions. In Sec.~\ref{sec:localityofnonstabilizerness}, our main results are presented: the main technical contribution, Theorem \ref{thm:arealawMDhardness}, shows that restricting MD states to be area-law still leads to computationally hard entanglement manipulation, and Theorem \ref{thm:main}, our main conceptual result, shows that the growth of magic in the subsystems separate area-law ED and MD states. In Sec.~\ref{sec:stabilizerregime}, we define the concept of a stabilizer regime, a restriction of the definition of a quantum phase that admits efficient entanglement manipulation and hosts area laws for entanglement and magic. We also study its stability and discuss how it connects with previous examples of Hamiltonian models discussed in the literature. Finally, in Sec.~\ref{sec:conclusions}, we summarize our main results and outline directions for future research.

\section{Phases of matter and scaling laws}
\label{sec:phasesandscaling}

In the following, we review existing results on the relationship between quantum phases and entanglement/magic scaling laws, while introducing the relevant definitions and notation.

\subsection{Quantum phases and the area law}
We will deal with the Hilbert space of $n$ qubits, $\mathcal H = (\mathbb C^2)^{\otimes n}$, where $n$ is always taken to be an asymptotic parameter, as we are interested in the scaling of quantities in the thermodynamic limit. We will make extensive use of Bachmann-Landau asymptotic notation\footnote{Given two functions $f,g$ over $\mathbb Z$, we denote: $f = O(g) \iff \exists c >0:  \lim_{n \to \infty }f(n)/g(n) \leq c$, $ f=\Omega(g) \iff \exists c>0: \lim_{n \to \infty} f(n)/g(n) \geq c$, introducing $f = \Theta(g)$ as a short-hand notation for  $f = O(g) \; \& \; f=\Omega(n)$. Furthermore, $f=o(g) \iff \lim_{n \to \infty}f(n)/g(n)=  0$ and $f=\omega(g) \iff \lim_{n \to \infty} f(n)/g(n) =\infty$.}, when referring to families of operators and states. We focus our attention on geometrically-local Hamiltonians:
\begin{equation}
    H = \sum_X h_X\;,
    \label{eq:Hamiltonian}
\end{equation}
where $\{h_X\}_X$ is a $\mathrm{poly}(n)$-size set with $X = \mathrm{supp}(h_X)$ of terms that are local with respect to the $D$-dimensional hypercubic lattice $\Lambda_D = \{1,\cdots,L\}^D$, meaning:
\begin{equation}
    \mathrm{diam}(X) \equiv \max_{i,j \in X}\mathrm{dist}_{\Lambda_D}(i,j) = O(1)\;,
\end{equation}
where $\mathrm{dist}_{\Lambda_D}(\cdot, \cdot)$ is the taxicab distance on $\Lambda_D$. In examples, we might also consider periodic boundary conditions. Furthermore, we assume that the family of hamiltonians in Eq. (\ref{eq:Hamiltonian}) satisfy a uniform energy bound:
\begin{equation}
    \sup_{n,i}\sum_{X\;:\; i \in X}\|h_X\| \leq J < \infty\;.
\end{equation}
We restrict ourselves to gapped Hamiltonians, whose ground state subspace is separated from the rest of the spectrum by a constant in the thermodynamic limit, referred to as the energy gap $\Delta_n  =\Omega(1)$. In this case, there is a clear-cut definition of a quantum phase based on the adiabatic continuity of quantum states \cite{hastings2005quasiadiabatic,bachmann2012automorphic}. More specifically, we say that two Hamiltonians $H$ and $H^\prime$ are in the same phase if there is a continuous path of (gapped) Hamiltonians that interpolate between the two. This is a generalization of the adiabatic theorem in quantum mechanics \cite{born1928beweis}, that states that eigenstates of Hamiltonians are preserved under perturbations if there is a energy gap. We require the gap to be non-vanishing in the thermodynamic limit in order to make the definition macroscopically robust.

We will mainly adopt a definition of quantum phases with a computational flavor. As argued in Ref.~\cite{chen2010local}, the adiabatic continuity condition can be reformulated in terms of finite-depth quantum circuits (FDQCs).
\begin{defn}
    We say that an $n$-qubit unitary $U$ is a geometrically local $D$-dimensional finite-depth quantum circuit if it admits the decomposition
    \begin{equation}
        U  =\prod_{\ell=1}^{\mathrm{depth}(U)} u_\ell\;,
    \end{equation}
    with $\mathrm{depth}(U)=O(1)$ and $u_\ell = \prod_{j} u_{j, \ell}$, $\{u_{j,\ell}\}_j$ defined on disjoint supports, $\max_{j,\ell}|\mathrm{supp}(u_{j,\ell})|=O(1)$ and $\max_{j,\ell}\mathrm{diam}(\mathrm{supp}(u_{j,\ell})) = O(1)$, where the diameter is computed on the hypercubic lattice of $D$-dimensions. Furthermore, we say that a (family of) two $n$ qubit states $\psi, \psi^\prime$ are in the same $D$-dimensional quantum phase if there is a $D$-dimensional FDQC $U$ such that $\psi^\prime = U \psi U^\dagger$.
    \label{def:fdqcphase}
\end{defn}
Sometimes we omit the dependency on $D$ and just state that $\psi$ and $\psi^\prime$ are in the same quantum phase. If a symmetry is imposed, the equivalence class further requires that gates commutes with a unitary representation of a group $G$. This definition can be derived from the former Hamiltonian formulation.\footnote{For a more detailed discussion, see \cite{coser2019classification} In general, it is shown that $|\psi^\prime\rangle \approx_{\varepsilon} U |\psi \rangle$, where $U$ is a quantum circuit with (i) finite depth and (ii) with gates acting at most $O(\mathrm{poly}\log(n))$ qubits, such that the error of approximation scales as $\varepsilon = \| U|\psi\rangle -|\psi^\prime \rangle \| \leq \delta_n$, where $\delta_n$ decays faster than any polynomial in $n$. We will not concern ourselves with such details, as it will suffice to take the definition with strictly local gates.} 
In the case of degenerate ground states, applying the FDQC to a basis of the ground-state subspace of $H$ generates a basis of the ground-state subspace of $H'$.

One reason this circuit definition is natural is that the FDQC only generates ``local'' entanglement. To see why, consider $A \subseteq [n]\equiv \{1,2,\cdots,n\}$. Given a pure state $|\phi\rangle$, with density matrix $\phi=|\phi\rangle\langle \phi|$, we denote the reduced density matrix of $A$ by  $\phi_A = \tr_{\bar A}\phi$, where $\bar A$ is the complement of $A$. The  $\alpha$-Rényi entropy of $A$ is given by \cite{renyi1961measures}
\begin{equation}
S_\alpha(\phi_A) \equiv \frac{\log_2 \tr(\phi_A^\alpha)}{1-\alpha}\;,
\label{eq:renyientropy}
\end{equation}
defined for $\alpha \in [0,\infty)$, where we will denote $S(\cdot) \equiv S_1(\cdot)$. It is widely conjectured that gapped ground states of geometrically local Hamiltonians do satisfy the area law, formalized as:
\begin{defn}
    We say that a quantum state $\psi$ is area-law if for $A \subseteq [n]$ if
    \begin{equation}
        S_\alpha(\psi_A) = O(|\partial A|)\;,
    \end{equation}
    for all $\alpha \in [0, \infty)$, where $\partial A \equiv\{j \in A\;|\; \mathrm{dist}_{\Lambda_D}(j,\bar A)=1\}$ is the number of sites at the boundary of $A$, with respect to some dimension $D$. If the state is area-law for all subregions, it is said that $\psi$ is area-law.
    \label{def:arealawstate}
\end{defn}
The von Neumann $(\alpha=1)$ entanglement area law is widely conjectured for ground states of uniformly gapped, geometrically local Hamiltonians with bounded interaction strength, under appropriate assumptions on ground-state degeneracy. It is established in one dimension \cite{hastings2007area,brandao2013area, arad2012improved, arad2013area}, while Ref. \cite{anshu2022area} proves a nearly linear entropy bound for locally gapped frustration-free two-dimensional systems across vertical bipartitions. Definition 2 imposes the stronger requirement of an area law also at $\alpha=0$, corresponding to an exact Schmidt-rank bound. We adopt this stronger condition as an explicit assumption; it is satisfied by the states discussed in this work.

As we prove in Appendix \ref{app:entropygrowthproof}, for large enough subsystems, the entropic growth under the FQDC is limited by the boundary size of the corresponding region.
 \begin{lem}
     Let $\psi = \ket{\psi}\bra{\psi}$ be a $D$-dimensional area-law state, and let $\psi^\prime = U \psi U^\dagger$ be in the same quantum phase. Then, $\psi^\prime$ is also area-law with the same $D$.
     \label{lem:entropygrowth}
 \end{lem}
The intuition is that finite-depth unitaries can only create entanglement in the region near the boundary, as schematically illustrated on the bottom of Fig. \ref{fig:theorem}. We say that FDQCs \emph{preserve} the area law by Lemma \ref{lem:entropygrowth}, since it maintains its entropy scaling. This reinforces FDQCs as the natural generator of gapped phases.  

However, we note that imposing area laws does not correspond to the space of ground states, as discussed in \cite{ge2016area}. There are area-law states that do not admit a parent geometrically-local Hamiltonian on which it is a ground state of. We will take the area-law as a necessary information-theoretic constraint, making our results relevant to other physical scenarios where the area law emerges beyond ground states, such as in monitored circuits \cite{vasseur2026houches}.

\subsection{Magic and volume-law typicality \label{sec:magictypicallity}}
The basics of stabilizer resource theory and non-stabilizerness are reviewed in App. \ref{app:stabreview}. We denote $\mathcal P_n$ as the Pauli group, $V_\mathcal S\subseteq \mathcal H$ as the stabilizer code of the stabilizer group $\mathcal S \subseteq \mathcal P_n$. The set of pure stabilizer states will be denoted as $\mathrm{stab}_n$, which is in bijection with the set of stabilizer groups of full rank. Given any $n$-qubit state $\psi$, one of the simplest measures is the stabilizer nullity, defined as \cite{beverland2020lower}
\begin{equation}
    \nu(\psi) = n-\mathrm{rank}[\mathrm{Stab}(\psi)]\;,
    \label{eq:stabnullity}
\end{equation}
where $\mathrm{Stab}(\psi) = \{P \in \mathcal P_n: |\tr(P \psi)|=1\}$ is its stabilizer group. For a pure stabilizer state, its stabilizer group is full rank, and thus has vanishing nullity as expected of a faithful monotone. $\nu$ counts how much of a quantum state is fixed by a stabilizer group, and as reviewed in App. \ref{app:structurenullity}, a pure quantum state of nullity $\nu$ admits can be compressed to \cite{leone2021quantum}
\begin{equation}
    \psi  = |\psi\rangle \langle \psi|\rightsquigarrow (\mathcal T_{\mathrm{Stab}(\psi)},\{h_i\}_{i=1}^k )\;,
    \label{eq:fixednullity_rep}
\end{equation}
where $\mathcal T_{\mathrm{Stab}(\psi)}$ is called the stabilizer tableau of the group $\mathrm{Stab}(\psi)$, corresponding to a binary matrix of its generators, and $\{h_i\}_{i=1}^k \subseteq \mathcal P_n$ is a set of Pauli operators with non-vanishing expectation value in the state. We also have that $k = O(4^{\nu})$: If $\nu = O(\log n)$, the corresponding quantum state can be classically stored with a polynomial amount of data. In \cite{gu2024magic}, it was shown that the 2-Rényi entropy of subsystems can be evaluated in this classical representation with a classical algorithm in  $O(4^\nu n^3)$ time, efficient in this sub-logarithm regime. 

For stabilizer states, it is known that the nullity corresponds to entanglement, and can be used to show the area law.
\begin{lem}[\cite{fattal2004entanglement, hamma2005bipartite}] 
    Let $\psi^\mathcal S \in \mathrm{stab}_n$ be a pure stabilizer state with stabilizer group $\mathcal S \subseteq \mathcal P_n$, and let $ A \subseteq [n]$. Then, for all $\alpha \in [0, \infty)$, we have that
    \begin{enumerate}
        \item $S_\alpha(\psi^\mathcal S_A) =\nu(\psi^\mathcal S_A) = |A|-\mathrm{rank}(\mathcal S_A)$, where $\mathcal{S}_A \equiv \{s \in \mathcal S\;:\;\tr_{\bar A}(s)\neq 0\}$ is the stabilizer subgroup supported on $A$; 
        \item Given $\mathcal S = \langle s_1, s_2, \cdots, s_n\rangle$ with the corresponding Hamiltonian $H_\mathcal S= -\sum_{i=1}^n s_i$ geometrically local in $D$ dimensions, then, there is a $c=O(1)$ constant such that:
        \begin{equation}
            S_\alpha (\psi^\mathcal S_A) \leq c |\partial A|\;.
        \end{equation}
        That is, $\psi^\mathcal S$ is area-law.
    \end{enumerate}
    \label{lem:stabarealaw}
\end{lem}
The proof is presented in App. \ref{app:stabarealaw} for completeness. As a measure of non-stabilizerness, the nullity is quite unstable: Since it only counts how many Paulis stabilize a quantum state, states that have a near-stabilizer Pauli spectrum, say by yielding $\tr (P\psi) = 1-\varepsilon$ for $n$ Paulis, will have maximum nullity for any $\varepsilon >0$. To capture such deformations of stabilizer states, it is useful to consider the Stabilizer Rényi entropy (SRE) \cite{leone2022stabilizer}, given by
\begin{equation}
    M_\alpha(\psi) \equiv \frac{1}{1-\alpha}\log_2\left[ \frac{1}{2^n} \sum_{P \in \mathcal P_n} \tr(P\psi)^{2\alpha}\right]\;,
    \label{eq:SRE}
\end{equation}
taking $\alpha \in [0, \infty)$. Both measures in Eqs. (\ref{eq:stabnullity}, \ref{eq:SRE}) are additive, positive, invariant under Clifford operators, and faithful. For a discussion of those properties, we defer to App. \ref{app:stabreview}. One important fact is that the SRE is upper-bounded by the nullity, $M_\alpha(\psi) \leq \nu(\psi)$, corroborating the intuition that it captures subtler deformations of stabilizer states.

Numerical studies \cite{haug2023quantifying, tarabunga2024nonstabilizerness} based on the stabilizer Rényi entropy suggest that, for generic ground states of one-dimensional Hamiltonians, the stabilizer Rényi entropy typically follows a volume law, that is,
\begin{equation}
    M_\alpha(\psi) \sim \#_1 n -\#_2+\cdots\;,
    \label{eq:magicvolumelaw}
\end{equation}
where $\#_1, \#_2$ are non-universal constants, found both in gapped and gapless phases. This seems to indicate that \emph{typical} ground states do have $M_\alpha(\psi) = \Theta(n)$. This extensive behaviour is the scaling obtained for typical random states \cite{leone2022stabilizer}, showing that ground states are typically as non-stabilizer as generic states. It is hard to make this statement precise, but based on work characterizing non-stabilizerness of random matrix product states \cite{chen2024magic}, it is possible to make a mathematical argument for this statement, as we put forward in App. \ref{app:vollawmpsargument}. This is somewhat expected: Since non-stabilizerness of a state is closely related to the $T$ gate count in its circuit synthesis \cite{beverland2020lower}, even if we have a short-range entangled ground state \cite{zeng2015gapped}, we do not expect that there is a non-trivial upper bound on the number of the $T$ gates necessary for its circuit, generically saturating it.

This discrepancy between the behaviors of entanglement and non-stabilizerness compared to the typical behavior suggests the need for an understanding of the scalings of entanglement and magic in area-law quantum states. This is what we will do in the next section.

\section{Locality of non-stabilizerness \label{sec:localityofnonstabilizerness}}
In this section, we show that by restricting the recent classification of quantum states to be area-law, we can show the presence of an analogue of area-law for non-stabilizerness for a certain class of states that also inherit efficiency in various information processing tasks, and show that states with faster-than-area growth of subsystem non-stabilizerness have computationally hard entanglement.

\subsection{Entanglement versus non-stabilizerness at large scales}
In \cite{gu2024magic}, the following definition was put forward:

\begin{defn}
    Given a region $A \subseteq [n]$ with $\min \{|A|, |\bar A|\} =\omega(\log n)$, we say that a family of $n$-qubit states $\psi =\ket{\psi}\bra{\psi}$ of stabilizer nullity $\nu$ is either:
    \begin{itemize}
        \item \textbf{Entanglement dominated} (ED), if $S(\psi_A) = \omega(\nu)$,
        \item \textbf{Magic dominated} (MD), if $S(\psi_A) = O(\nu)$.
    \end{itemize}
    The ED/MD classification depends on the chosen family of bipartitions $A|\bar A$. Unless otherwise stated, ED and MD refer to this specified family of cuts.
    \label{def:EDMDstate}
\end{defn}
Note that this definition applies for large subregions, where the corresponding reduced density matrix $\psi_A = \tr_{\bar A}\psi$ does not have an efficient compression in general, since they are of dimension $2^{\omega(\log n)}$. Entanglement-dominated states are characterized by states that have asymptotically low nullity compared to their entanglement, and therefore have near-stabilizer properties. Magic-dominated states, on the other hand, have high nullity compared to their entanglement and have no particular stabilizer structure.

Based on this dichotomy, \cite{gu2024magic} was able to characterize the \emph{hardness} of entanglement manipulation in the two classes of states. Consider a subregion $A \subseteq [n]$, and consider two observers, one in $A$ and another in the complement $\bar A$, having access to $m$ samples of a state $\psi$, and access to polynomial-time quantum circuits, classical communication between $A$ and $\bar A$, and computational-basis measurements. Three tasks were classified, illustrated in Fig. \ref{fig:EDentanglementprocessing}. First, we have entanglement estimation, which aims to compute $S(\psi_A)$ up to some $\varepsilon$ relative error. Second, entanglement distillation aims to distill $M_+$ Bell pairs from $\psi$ using LOCC operations. Lastly, entanglement dilution is the task of preparing $\psi$ from $M_-$ Bell pairs across the $A|\bar A$ partition constrained to $N_\mathrm{CC}$ bits of communication. It was shown that ED-dominated states have efficient protocols for all three tasks, where entanglement can be estimated within $\varepsilon = o(1)$ relative error, the number of Bell pairs that can be: (i) extracted, with $M_+/S(\psi_A)  = 1-o(1)$ and (ii) diluted, where $M_-/S(\psi_A) =1+o(1)$ are nearly optimal. For MD states, it was shown that efficient protocols with such parameters are impossible.  

We aim to understand this classification when we restrict to area-law states. In what follows, we will denote an area-law MD/ED state as an area-law state that is ED/MD for all choices of subregions $A$ in Definition \ref{def:EDMDstate}. For MD states, we show that this restriction does not change the hardness of the entanglement manipulation tasks.
\begin{thm}
    There is no single state-agnostic, sample-efficient protocol that achieves the following guarantees uniformly over all area-law MD states satisfying \(S(\psi_A)=\omega(\log n)\):  (1) estimates $S(\psi_A)$ with $\leq 1/3$ relative error, (2) distills a vanishing fraction of Bell pairs, with $M_+/S(\psi_A) \geq c >0$ and (3) dilutes $\psi$ across $A|\bar A$ using  $M_-/S(\psi_A) \leq C < \infty$ Bell pairs, with $c,C$ being fixed constants.
    \label{thm:arealawMDhardness}
\end{thm}

The proof is presented in App. \ref{app:arealawMDhardness}. The proof technique relies in a careful revisit of the arguments of \cite{gu2024magic}, showing that certain ensembles of MD states can be embedded into lattices using sphere-packings. This indicates that even states that host an entanglement area-law have hard entanglement manipulation. This is in line with previous work establishing the hardness of optimizing, preparing, and realizing area-law states for $D \geq 2$, due to their intrinsic exponential dimension \cite{ge2016area}. Our results show that if their magic is extensive, they also carry an intrinsic hardness to manipulate. We can readily apply this result to our discussion of ground states, since an extensive SRE is sufficient to witness MD states.
\begin{prop}
    Let $\psi$ be a quantum state with extensive SRE for some $\alpha$, $M_\alpha(\psi) = \Theta(n)$. Then, $\psi$ is magic-dominated for any subregion $A$.
\end{prop}
\begin{proof}
    The family of SRE serves as a lower bound for the stabilizer nullity \cite{leone2022stabilizer} $\nu(\psi) \geq M_\alpha(\psi)$, implying that $\nu = \Theta(n)$.  Since $S(\psi_A) \leq n$, for any $A$, it follows that $S(\psi_A) = O(\nu)$.
\end{proof}

\begin{figure*}[t!]
    \centering
    \includegraphics[width=0.6\linewidth]{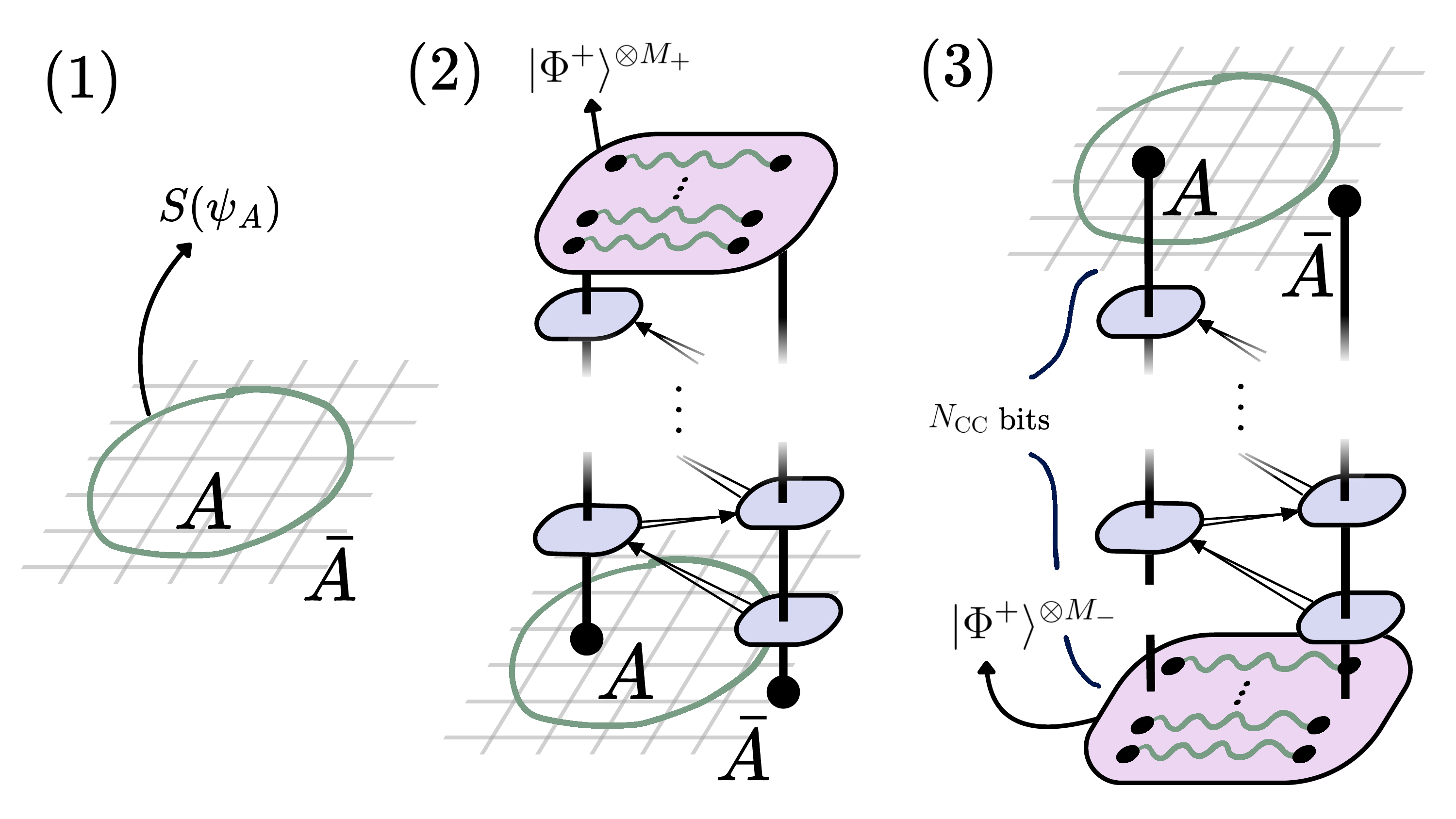}
    \caption{The three information-theoretic tasks on which entanglement-dominated states admit quantum algorithms that are both time- and sample-efficient, while magic-dominated states, even if they are area-law, do not.  (1) entanglement estimation, with the task of estimating the entanglement entropy between the bipartition defined by the region; (2) distillation, which seeks to extract Bell pairs from LOCC operations on the bipartition; (3) dilution, where the task is to build the target state from Bell pairs under LOCC with bounded communication complexity.}
    \label{fig:EDentanglementprocessing}
\end{figure*}

By Theorem \ref{thm:arealawMDhardness} and the discussion of Sec. \ref{sec:phasesandscaling}, this suggests that ground states are usually magic-dominated, making them a poor resource for entanglement manipulation under computational constraints.

An interesting point to make is that $D=1$ apparently avoids the hardness result. Indeed, as pointed out in \cite{gu2024magic}, MPS are magic-dominated and are well-known to efficiently approximate one-dimensional gapped ground states \cite{hastings2007area}. MPS-specific efficient entanglement protocols cannot be ruled out: for example, we know that they admit efficient tomography \cite{cramer2010efficient}. Therefore, $D=1$ area-law states can avoid the hardness barrier imposed in the separation of ED and MD states.

\subsection{Local SRE \& the area-law}

Having established the hardness result for area-law states, we now show that this transition between the area-law ED states to MD states is captured by an abrupt change in the behaviour of the magic of subsystems. To measure it, we refer to the following measure \cite{leone2022stabilizer, esposito2026stabilizer}:
\begin{equation}
    \tilde M_2(\psi) = M_2(\psi)-S_2(\psi)\;,
    \label{eq:localstab}
\end{equation}
which can be shown to vanish if and only if $\psi \propto \sum_{P \in \mathcal S} P$, where $P$ is not a necessarily full rank stabilizer group $\mathcal S \subseteq \mathcal P_n$. This is the case for subsystems of a pure stabilizer state $\psi^\mathcal S$, and thus we have that $\tilde M_2(\psi^\mathcal S_A)=0$ for any $A \subseteq[n]$, as reviewed in App. \ref{app:stabreview}. If we restrict to stabilizer pure states, $M_2$ is already faithful, and there is no need to consider this entropy-subtracted version.

Recall that states of fixed nullity have a compressed description described in Eq. (\ref{eq:fixednullity}). We can further show that the estimation of the stabilizer Rényi entropy for nullity restricted states has the following form:
\begin{lem}
    Let $\psi = \ket{\psi}\bra \psi$ be a $n$ qubit pure state with fixed nullity $\nu$ admitting a classical representation described on Eq. (\ref{eq:fixednullity_rep}), for any $A \subseteq [n]$ it follows that
    \begin{equation}
        \tilde M_2(\psi_A) = - \log\left[\frac{\sum_{i=0}^k \delta_i \langle h_i \rangle_\psi^{4}}{ \sum_{i=0}^k \delta_i \langle h_i \rangle_\psi^2}\right]\;,
    \end{equation}
    where $\delta_i =1$ if there is a stabilizer $Q \in \mathrm{Stab}(\psi)$ such that $\mathrm{supp}(h_i Q) \subseteq A$ and $\delta_i =0$ otherwise. Furthermore, there is a classical algorithm that takes as an input the representation and computes $\tilde M_2(\psi_A)$ in $O(4^\nu n^3)$ time.
    \label{lem:stab_renyi_formula}
\end{lem}
The proof is presented in App. \ref{app:stab_renyi_formula}, and shows that states of fixed nullity also have their local non-stabilizerness obtained within the same complexity as their subsystem entropy.  It also shows an interpretation: It measures the fluctuations of the bad generators $\{h_i\}$, responsible for generating the non-stabilizerness itself.

The technical result that allows us to connect $\tilde M_2$ with ED and MD states is the following upper bound:
\begin{lem}
    For any $n$ qubit state $\psi$, $\tilde M_2(\psi_A) \leq \nu(\psi) \;.$
\label{lem:localstabbound}
\end{lem}
The proof is presented in App. \ref{app:localstabbound}. From this result and Definition \ref{def:EDMDstate}, it follows directly that:
\begin{thm}
Let $\psi$ be an area-law state. Then, if it is ED, $\tilde M_2(\psi_A) =O(|\partial A|)$. On the other hand, $\psi$ is MD if $\tilde M_2(\psi_A) = \omega(|\partial A|)$.
\label{thm:main}
\end{thm}
This shows our second main claim, as the arrows indicate in Fig. \ref{fig:theorem}: It displays that entanglement area-law ground states do display non-stabilizer area laws if they are entanglement dominated, and magic-dominated states can be witnessed by their faster-than-area growth. This result establishes an operational meaning both to the measure  $\tilde M_2$ in the setting of subsystems, as a witness to the ED/MD separation, and for a ``magic area law'', that is manifested as the signature of entanglement domination.

\section{Stabilizer regime, area-laws and examples \label{sec:stabilizerregime}}
We have then established that area-law ED states also hosts a non-stabilizerness area law, it is natural to ask from what perturbations this is a robust feature, as entanglement area laws are for FDQC perturbations. We know that FDQCs are too strong a class, since they can generate arbitrary magic. This motivates us to restrict Definition \ref{def:fdqcphase} as:
\begin{defn}
    We say that a $n$-qubit unitary $U$ is $t$-doped if there is $t$ non-Clifford gates. We say that two $n$ qubit area-law states $\psi, \psi^\prime$ are in the same stabilizer regime if there is a $t=o(\log n)$-doped circuit $U \in \mathrm{FDQC}_D$ such that $\psi = U\psi^\prime U^\dagger$ that preserves entropy scaling, $S(\psi^\prime_A) = \Theta(S(\psi_A))$. 
    \label{def:stabilizerregime}
\end{defn}
We claim that this is the natural restriction that preserves the magic area-law and the ED phase. The natural setting is to consider an area-law $\psi^\mathcal S \in \mathrm{stab}_n$, generated through geometrically-local stabilizer Hamiltonians as explained in Lemma \ref{lem:stabarealaw}. Consider the new ground state generated by adding a perturbation, such that
\begin{equation}
    H_\mathcal S \mapsto H_\mathcal S+ \sum_{k=1}^p \lambda_k P_k\;,
\end{equation}
where $\{P_k\}\subseteq \mathcal P_n$ is a subset of Pauli strings. In \cite{gu2024doped}, it was shown that the corresponding ground state has nullity $\nu \leq p$. Since, as a consequence of Proposition  \ref{prop:stabnullitybound} in App. \ref{app:stabregimerobust}, it would require a $\Omega(p)$-doped circuit to be synthesized from a stabilizer state. For the case where $p=o(\log n)$, the definition of a stabilizer regime captures the minimal regime when the lower bound is attained, and also imposes finite-depth to model the local changes to the ground state. 

We can then show a version of Lemma \ref{lem:entropygrowth}, for $\tilde M_2$:
\begin{lem}
    Given $A\subseteq [n]$, let $\psi$ be a ED area-law state with $\psi^\prime$ in the same stabilizer regime, assuming $S(\psi_A) = \omega(\log n)$. Then, $\psi^\prime$ is also ED and area-law.
    \label{lem:stabregimerobust}
\end{lem}
Shown in App. \ref{app:stabregimerobust}. By Theorem \ref{thm:main}, this also shows that the area law for $\tilde M_2$ is preserved. It shows that, given some state $\psi$ that is ED and area-law (as, for example, a stabilizer ground state, discussed in the next section), there is a restricted notion of a quantum phase that preserves its structured entanglement, that can also be efficiently manipulated. This establishes the importance of area laws in non-stabilizerness as a distinct signature of efficient entanglement in phases of matter.

We note that the asymptotic divide between the maximum number of non-Clifford gates $t$ in Definition \ref{def:stabilizerregime} and the minimum growth of the entanglement entropy in Lemma \ref{lem:stabregimerobust} is somewhat arbitrary, and we made the minimum choice such that stabilizer regimes are defined with a doping parameter that can depend on $n$. In the case $t=O(1)$, having $S(\psi_A)=\omega(1)$ entanglement suffices to show the result, as the proof discusses. Another point is the imposition that the entropic scaling must be invariant under the action of the unitary. One can also note the condition for the entropy scaling, which precludes us from considering unitaries that disentangle the initial state, making $\psi^\prime$ a product state. In this case, we can have MD states very easily, by just considering a product of non-stabilizer states.

\subsection{Stabilizer ground states and ED to MD transition}
We close this section by making some remarks about the physical applications of our results, while providing specific examples. We recall that a geometrically-local $H_\mathcal S$ has as its ground state the stabilizer \emph{code}:
\begin{equation}
    V_\mathcal S = \{\ket{\psi} \in \mathcal H\;|\; s \ket{\psi} =+\ket{\psi}\;:\;\forall s \in \mathcal S\}\;.
\end{equation}
Two parameters are relevant: The \emph{encoding rate} $k=\log_2 \mathrm{dim}(V_\mathcal S)$ counts the ground space degeneracy, and counts the effective number of qubits/bits that can be encoded on it. The second parameter measures its robustness: To define it, let $\mathrm C(\mathcal S) \subseteq \mathcal  P_n$ be the Pauli commutant with respect to $\mathcal S$. The \emph{code distance} $d$ is defined as:
\begin{equation}
    d =\min_{P \in \mathrm C(\mathcal S)\backslash \mathcal S} |\mathrm{supp}(P)|\;. 
    \label{eq:distancedef}
\end{equation}
It is interpreted as the minimum support of a Pauli operator that fails to be detected on the code. If $d=O(1)$, local operators already cannot be detected, and thus  $V_\mathcal S$ is not a performing quantum code in the asymptotic regime, and the corresponding stabilizer ground state is not stable under Hamiltonian perturbations \cite{brayvi2010topological}. Otherwise, if the distance is a non-trivial function of $n$, it means that local Paulis can be detected, also implying robustness to local perturbations. We will refer to $V_\mathcal S$ (or $\mathcal S$) as  a $[[n,k,d]]$ code. If one is interested also only in symmetric perturbations with respect to some unitary group $G$, the optimization over the commutant can be restricted to symmetric Paulis obtain a symmetric distance $d_G$ \cite{zeng2019quantum}. Some canonical examples in the many-body literature include:
\begin{enumerate}
    \item \textbf{Ising model:} Defined in $D$ dimensions, $\mathcal S_{\mathrm{ising},D}\equiv \langle \{Z_i Z_j\}_{\langle i, j\rangle}\rangle$, where $\langle i,j\rangle$ denotes nearest-neighbors in the $D$-dimensional hypercubic lattice. It stabilizes $V_{\mathcal S_{\mathrm{ising},D}}\equiv \mathrm{span}_\mathbb C\{\ket{0}^{\otimes n}, \ket{1}^{\otimes n}\}$. It has $d=1$, since any $Z_i $ is feasible in Eq. (\ref{eq:distancedef}), but if a $G=\mathbb Z_2$ bitflip symmetry is imposed, we have $d_{\mathbb Z_2} =\Theta(n)$. Due to this reason, it is referred as a classical code, encoding a $k=1$ bit.
    \item  \textbf{Cluster model:} On its simplest version, it is defined on $D=1$ \cite{briegel2001persistent} as $\mathcal S_{\mathrm{cluster},1}\equiv \langle \{Z_i X_{i+1}Z_{i+2}\}_i\rangle$, but has also higher-dimensional generalizations \cite{raussendorf2005long, else2012symmetry}. In this model, it also hosts a logical $k=1$ bit (it also has $d=1$) on its code, but its spanning states are entangled \cite{zeng2019quantum}. Correspondingly, it also has symmetry protection, with $d_{\mathbb Z_2 \times \mathbb Z_2}=\Theta(n)$.
    \item \textbf{Toric code: }Defined by $\mathcal S_{\mathrm{TC},2} =\langle \{XXXX\}_{\square}, \{ZZZZ\}_+ \rangle$, where $\square$ and $+$ denotes a plaquette and a vertex of a two-dimensional square lattice, tessellating a torus, where the qubits are defined on the edges \cite{kitaev2003fault}. It has $k=2$ logical qubits. It has a large distance, $d=\Theta(\sqrt{n})$, which is, in fact, the asymptotic limit for the distance in geometrically-local $D=2$ codes \cite{bravyi2010tradeoffs}. More robust, higher $D$ versions of the toric code, which are even more robust can also be constructed \cite{dennis2002topological}.
\end{enumerate}
In all those examples, there is some degree of robustness of the ground state subspace that allow us to detect errors. This allows to show that the corresponding ground space is robust upon adding generic perturbations in the Hamiltonian \cite{brayvi2010topological}. However, on general grounds, this do not guarantee that their entanglement is efficiently manipulable, since the resulting state might be MD, and also usually hosts a magic volume-law.

On a much studied family of examples, consider, given $H_\mathcal S$:
\begin{equation}
    H_\mathcal S(h_x,h_z;p) = H_\mathcal S + \sum_{i=1}^p(\lambda_x X_i + \lambda_z Z_i)\;,
    \label{eq:perturb_stabham}
\end{equation}
for $p=n$, this is the usual transverse/longitudinal field models that are now textbook \cite{zeng2019quantum}\footnote{ For the symmetry-restricted cases, imposing the symmetry usually implies either $\lambda_x=0$ or $\lambda_z=0$}. For sufficiently low $|\lambda_x|, |\lambda_z| \ll 1$, it is indeed numerically evident that the area law and the robust ground-state subspace exist, and that there is extensive magic in the $D=1$ case \cite{tarabunga2024nonstabilizerness}. Our results, however, show that if one takes $p= o(\log n)$, a generic point in the phase diagram should exhibit an area-law for the non-stabilizerness, a sign that the quantum state has also efficient processing. On a intermediate asymptotic regime, we expect to see a MD ground space marked by faster-than-area growth in the non-stabilizerness of the subsystems, while we maintain the entanglement area-law. This can be used as a probe of a possible ED to MD transition, and to understand what ground states admit efficient entanglement manipulation.

\section{Conclusions \label{sec:conclusions}}
In this work, we investigated the interplay between entanglement and non-stabilizerness in quantum many-body systems. While ground states of locally interacting Hamiltonians are known to obey an entanglement area-law, they can nonetheless exhibit volume-law scaling of magic, as quantified by the stabilizer Rényi entropy. This observation highlights that entanglement alone does not fully characterize the complexity of such states and that magic plays a crucial and often overlooked role in determining their computational and informational properties.

To resolve the apparent tension between the presence of stabilizer code structures and the fragility of stabilizerness under perturbations, we adopted the recently proposed entanglement-dominated (ED) versus magic-dominated (MD) classification of Hilbert space. Our main results demonstrate that this ED/MD dichotomy persists even within the set of area-law states, where MD states exhibit  provably hard entanglement manipulation, whereas ED states admit general efficient protocols (Theorem \ref{thm:arealawMDhardness}). This is captured by the behaviour of their non-stabilizerness of subsystems: states in the ED regime obey a non-stabilizerness area-laws, whereas fast-than-area magic growth witnesses MD states (Theorem \ref{thm:main}). Based on this separation, we defined the stabilizer regime, a restriction of the notion of a quantum phase that preserves magic area-laws and allow efficient entanglement processing.

We further showed that ED states can naturally arise in stabilizer codes subject to specific perturbations, where the system remains in the same quantum phase while transitioning between MD and ED behavious. This finding offers new insight into the structure of topological codes like the toric code, suggesting that their physical resilience under perturbations must be understood not only in terms of phase stability, but also through the lens of of magic stability, given its relevance for decoding and information retrieval.

Looking forward, several directions remain open. First, extending this non-stabilizerness area versus volume law classification to mixed states, such as thermal or noisy ground states, is of direct relevance for quantum error correction, particularly in light of recent advances connecting phases of matter and decoding thresholds \cite{sang2024mixed}. Second, our results hint at deeper connections with measurement-based quantum computation (MBQC) \cite{raussendorf2001one}, where non-stabilizerness is known to play an important role \cite{liu2022many}. Furthermore, it would be interesting to perform a numerical analysis of the entanglement manipulation tasks in the ground states of Eq. (\ref{eq:perturb_stabham}), to determine whether the ED/MD transition correlates with the algorithmic complexity of these operations.

\begin{acknowledgments}
We acknowledge the use of ChatGPT (versions 5.6 and 6) for assistance in verifying specific arguments in the proofs. Most of the conceptual work was done prior to the availability of these models \cite{macedoquantum}. We thank Rodrigo Pereira and Alioscia Hamma for discussions and collaboration on an earlier version of the project.
R.A.M acknowledges funding from Coordenação de Aperfeiçoamento de Pessoal de Nível Superior – Brasil (CAPES) – Finance Code 001, and the ERC grant GIFNEQ 101163938.
R.C. acknowledges funding the Brazilian National Council for Scientific and Technological Development (CNPq; Grants No. 403181/2024-0 and 301687/2025-0), the National Institute of Science and Technology for Applied Quantum Computing through CNPq Process No. 408884/2024-0, the Simons Foundation (Grant No. 1023171, R.C.), and the Financiadora de Estudos e Projetos (Grant No. 1699/24 IIF-FINEP). R.C. also thanks the Technical University of Denmark for its hospitality, where part of this work was carried out during a guest professorship supported by the Otto M\o nsted Foundation. 
\end{acknowledgments}

\bibliographystyle{apsrev4-2}
\bibliography{pr_version.bib} 

\appendix
\onecolumngrid

\section{A brief review of stabilizer resource theory and non-stabilizer measures \label{app:stabreview}}
We will review the basics of stabilizer resource theory \cite{gottesman1997stabilizer, gottesman2016surviving, veitch2014resource}, by highlighting the concepts that will be necessary for our work. We will review the relevant results citing the corresponding original references, with proofs presented when they add to the discussion. A stabilizer group $\mathcal S \subseteq \mathcal P_n$ is an abelian subgroup of the Pauli group $\mathcal  P_n$ with $-1 \notin \mathcal S$. We always denote $\psi^\mathcal S$ as:
\begin{equation}
    \psi^\mathcal S \equiv \frac{1}{2^n}\sum_{s \in \mathcal S}s \in \mathcal S(\mathcal H)\;.
\end{equation}
It defines a stabilizer \emph{code}, defined as:
\begin{equation}
    V_\mathcal S \equiv\{\ket{\psi} \in \mathcal H\;|\; s\ket{\psi} = +\ket{\psi}\;,\;\forall s \in \mathcal S\}\;.
\end{equation}
Its dimension can be read by its encoding rate, defined as $k=\log_2 \mathrm{dim}V_\mathcal S =n-\mathrm{rank}(\mathcal S)$. Note that if it is full rank, $\mathrm{rank}(\mathcal S)= n$, it stabilizes a single pure state $\psi^\mathcal S= \ket{\psi^\mathcal S}\bra{\psi^\mathcal S}$. There are three sets of ``stabilizer states'':
\begin{defn}
    Let $n$ be the number of qubits. We define the set of:
    \begin{align}
         &\mathrm{stab}_n \equiv \{\psi^\mathcal S \;|\; \mathcal S\subseteq \mathcal P_n \;\mathrm{stabilizer}\;\&\;\mathrm{rank}(\mathcal S)=n\}\;,\\
        &\tilde{\mathrm{stab}}_n \equiv \{\psi^\mathcal S\;|\;\mathcal S\subseteq \mathcal P_n \;\mathrm{stabilizer}\}\;,\\
        &\mathrm{STAB}_n \equiv \mathrm{conv}(\mathrm{stab}_n)=\{\sum_{\mathcal S \in\mathrm{stab}_n}p_\mathcal S \psi^\mathcal S\;|\; p_\mathcal S \geq 0\;,\; \sum_{\mathcal S \in \mathrm{stab}_n}p_\mathcal S=1\}\;,
    \end{align}
    note the hierarchy $\mathrm{stab}_n \subseteq \tilde{\mathrm{stab}}_n \subseteq \mathrm{STAB}_n$, which is strict for $n>1$. We will refer to $\mathrm{stab}_n$ as both the set of pure stabilizer states and (full rank) stabilizer groups, and $\mathrm{STAB}_n$ as the stabilizer polytope.
\end{defn}
Note that all three sets are Clifford invariant, meaning that if a state $\psi$ is a member of any of the three sets, so does $C\psi C^\dagger$ for any Clifford unitary $C$. With that, we define:
\begin{defn}
    Let $\mathcal C \ell_n $ denote the set of $n$ qubit Clifford unitaries. Given $\mathcal F \in \{\mathrm{stab}_n, \tilde{\mathrm{stab}}_n, \mathrm{STAB}_n\}$, we define the corresponding stabilizer resource theory as the tuple $(\mathcal F, \mathcal C \ell_n)$. We say that a function of quantum states $f:\mathcal S(\mathcal H) \to \mathbb R_{\geq 0}$ is a stabilizer monotone if $f(C\psi C^\dagger)\leq f(\psi)$ for all states $\psi$ and Cliffords $C \in \mathcal C \ell_n$ and is furthermore faithful if $f(\psi) =0 \Leftrightarrow \psi \in \mathcal F$.
\end{defn}

The three different sets corresponds to three different possible free state sets in the resource theory. We will not need to be concerned with the full mixed stabilizer state resource theory, described by $\mathrm{STAB}_n$: We will be mostly concerned with the cases related to $\mathrm{stab}_n, \tilde{\mathrm{stab}}_n$. Here is a Proposition that organizes all the properties of the monotones of interest to us:
\begin{prop}
(\cite{leone2022stabilizer, leone2024stabilizer, beverland2020lower}) Let $\psi$ be a $n$ qubit quantum state. Define the following measures:
\begin{equation}
    M_\alpha(\psi) \equiv \frac{\log\left[\frac{1}{2^n}\sum_{P \in \mathcal P_n}\tr(P\psi)^{2\alpha}\right]}{1-\alpha} \quad , \quad \tilde M_2(\psi)\equiv M_2(\psi)-S_2(\psi) \quad , \quad \nu(\psi) \equiv n-\mathrm{rank}[\mathrm{Stab}(\psi)]\;,
\end{equation}
where we take $\alpha \in [0,\infty)$ and we defined $\mathrm{Stab}(\psi) \equiv \{P \in \mathcal P_n\;|\;|\mathrm{tr}(P\psi)|=1\}$. We will refer to $f \in \{M_\alpha, \tilde M_2, \nu\}$ when denoting a property valid for all three measures. The following properties hold:
\begin{enumerate}
    \item \textbf{Faithfulness: }All of them are positive-semidefinite, and are bounded above by $\leq n$, satisfying: (i) $M_\alpha(\psi) =0 \;\mathrm{or}\;\nu(\psi) = 0 \Leftrightarrow \psi \in \mathrm{stab}_n$, (ii) $\tilde M_\alpha(\psi) =0 \Leftrightarrow \psi \in \tilde{\mathrm{stab}}_n$;
    \item \textbf{Monotonicity: }All measures are Clifford-invariant, $f(C\psi C^\dagger) = f(\psi)$ for all Clifford unitaries $C$;
    \item They are additive: For any other $m$ qubit state $\sigma$, $f(\psi\otimes \sigma ) = f(\psi) + f(\sigma)$;
    \item The Rényi entropies forms an hierarchy, $M_\alpha(\psi) < M_{\alpha^\prime}(\psi)$ for any $\psi$ if $\alpha > \alpha^\prime$. Furthermore, for pure states, $M_\alpha(\psi)\leq \nu(\psi)$.
    \item The SRE lower bounds the $T$ gate count of pure states: If $t(\psi)$ is the minimum number of $T$ gates to prepare $\psi = \ket{\psi}\bra{\psi}$ from a stabilizer state, $M_\alpha(\psi)\leq t(\psi)$. 
\end{enumerate}
\label{prop:properties_monotomes}
\end{prop}
Some further properties are listed in the Supplemental Material of \cite{gu2024pseudomagic}. Notice that the definition of $\tilde M_2$ is such to construct a faithful measure for the $\tilde{\mathrm{stab}}_n$ resource theory. It captures subsystem non-stabilizerness, as it is easily shown:
\begin{prop}
    Let $\psi^\mathcal S  \in \tilde{\mathrm{stab}}_n$ and $A \subseteq [n]$. Then, 
    \begin{equation}
        \psi_A^\mathcal S = \tr_{\bar A}\psi^\mathcal S =\frac{1}{2^{|A|}}\sum_{s_A\in \mathcal S_A}s_A \in \tilde{\mathrm{stab}}_n\;,
    \end{equation}
    with $\mathcal S_A$ as defined on part 1 of Lemma \ref{lem:stabarealaw}.
    \label{prop:reducedstab}
\end{prop}
\begin{proof}
    Note that:
    \begin{align}
        \psi^\mathcal S_A &= \frac{1}{2^n}\sum_{s\in \mathcal S}\tr_{\bar A}(s) = \frac{1}{2^{|A|}}\sum_{s\in \mathcal S_A}s_A\;,
    \end{align}
    where on the second equality we have used the fact that if $\tr_{\bar A}s\neq 0$, it must be of the form $s =s_{A}\otimes \mathbb I_{\bar A}$, since otherwise the trace vanishes. Membership in $\tilde{\mathrm{stab}}_n$ is immediate by its definition.
\end{proof}
\noindent Hence, $\tilde M_2(\psi_A) >0$ witness $\psi_A\notin \tilde{\mathrm{stab}}_n$, by property 1 of Proposition \ref{prop:properties_monotomes}. It measures how far a subsystem is from coming from a stabilizer state.

\subsection{The structure of fixed nullity states \label{app:structurenullity}}

Since the nullity counts the number of Paulis fixed by a stabilizer group, we expect that given a $n$ qubit quantum state:
\begin{equation}
    \psi = \frac{1}{2^n}
    \sum_{P \in \mathcal P_n} \langle P\rangle_\psi P
\end{equation}
to have its Pauli representation constrained by its nullity. Indeed, this is the result shown in previous works, for pure states:
\begin{lem}
(\cite{leone2024learning,gu2024magic})
    Let $\psi = \ket{\psi}\bra{\psi}$ be $n$ qubit pure quantum state with nullity $\nu(\psi) =\nu$. Then, it admits the following presentation:
    \begin{equation}
        \psi= K\cdot \Pi\;,
        \label{eq:fixednullity}
    \end{equation}
    where:
    \begin{align}
    K & =\frac{1}{2^\nu} \sum_{i=0}^k \tr(h_i\psi) h_i, \\
    \Pi & =\frac{1}{|\operatorname{Stab}(\psi)|} \sum_{P \in \operatorname{Stab}(\psi)} P,
    \end{align}
    with $\Pi$ being the projector into $\mathrm{Stab}(\psi)$ and the set $\{h_i\}_{i=0}^k$ is referred as the set of bad generators, which are Paulis such that $\mathrm{supp}(\psi)\equiv \{P \in\mathcal P_n\;|\;\tr(h_i \psi)\neq 0\}$ splits into left cosets of the stabilizer group, 
    \begin{equation}
        \mathrm{supp}(\psi)  = \bigsqcup_{i=0}^k h_i \mathrm{Stab}(\psi)=\bigsqcup_{i=0}^k \{h_i s \;|\; s \in \mathrm{Stab}(\psi)\}\;,
        \label{eq:cosetdecomp}
    \end{equation}
    with $h_0=\mathbb I$. Furthermore, $[K, \Pi] =0$, $k+1 \leq 4^\nu$ and $\|K\|\leq 2^\nu$.
    \label{lem:fixednullity}
\end{lem}
This Lemma shows how the knowledge of the stabilizer group $\mathrm{Stab}(\psi)$ is a strong constraint on its Pauli support: It reduces the need of tomographic information of $\psi$ to determining $O(4^\nu)$ amplitudes in $K$.

Now, we can construct a classical representation of the quantum state. To do so, it is convenient to take a representation of Pauli operators in terms of bitstrings, dubbed what is called the symplectic representation of Pauli operators \cite{gottesman1997stabilizer, gottesman1998heisenberg}, where given $P \in \mathcal P_n$, we can write it as:
\begin{equation}
    P = \lambda X^\mathbf a Z^\mathbf b\;,
    \label{eq:symppauli}
\end{equation}
where $\lambda \in \{\pm 1, \pm i\}$ and $\mathbf a, \mathbf b \in \mathbb F_2^n$ are binary vectors, where we are using the shorthand $Q^\mathbf \equiv Q^{x_1}\otimes Q^{x_2}\otimes \cdots \otimes Q^{x_n}$. Note that storing $P \mapsto \{\lambda, \mathbf a,\mathbf b\}$ takes $O(n)$ bits, and Pauli multiplication, that takes $P_1 \in \mathcal P_n \mapsto \{\lambda_1, \mathbf a_1,\mathbf b_1\}$ and $P_2\in \mathcal P_n\mapsto \{\lambda_2, \mathbf a_2,\mathbf b_2\}$ and computes the representation of $P_1\cdot P_2 \mapsto \{\lambda_1\lambda_2 (-1)^{\mathbf b_1 \cdot \mathbf a_2}, \mathbf a_1 +\mathbf a_2, \mathbf b_1 +\mathbf b_2\}$, takes $O(n)$ time, and thus checking commutativity of two Pauli strings. Given $P \in \mathcal P_n$, denote $\mathrm S(P) = (\mathbf a , \mathbf b)^T\in \mathbb F_2^{2n}$ as its binary vector on the corresponding symplectic representation. 

There is a very nice linear-algebraic characterization of stabilizer groups in this language that will be useful later:
\begin{prop}
    Let $\mathcal S\subseteq \mathcal P_n$ be a stabilizer group, and denote $\mathrm{gen}(\mathcal S) = \{g_1, g_2, \cdots, g_{\mathrm{rank}(\mathcal S)}\}$ as a minimal generating set. We define its stabilizer tableau $\mathcal T_\mathcal S$ as:
    \begin{equation}
        \mathcal T_\mathcal S = \begin{pmatrix}
            \mathrm S(g_1) \\
            \mathrm S(g_3)\\
            \vdots\\
            \mathrm S(g_{\mathrm{rank}(\mathcal S)})
        \end{pmatrix}
        \in \mathbb F_2^{\mathrm{rank}(\mathcal S)\times 2n}
    \end{equation}
    Then:
    \begin{equation}
        s \in \mathcal S \Rightarrow \mathrm S(s) \in \mathrm{rowspace}(\mathcal T_\mathcal S) \quad \& \quad \mathbf v\in \mathrm{rowspace}(\mathcal T_\mathcal S) \Rightarrow \exists s \in \mathcal S: \mathrm S(s) =\mathbf v\;.
    \end{equation}
    \label{prop:stabmembership}
\end{prop}
\begin{proof}
 If $s\in \mathcal S$, then it exists a $\mathbf c \in \mathbb F_2^{\mathrm{rank}(\mathcal S)}$ such that:
    \begin{equation}
        s = \prod_{j=1}^{\mathrm{rank}(\mathcal S)}g_j ^{c_j}\Rightarrow \mathrm S(s) = \sum_{j=1}^{\mathrm{rank}(\mathcal S)}c_j \mathrm S(g_j)\in \mathrm{rowspace}(\mathcal T_\mathcal S)\;.
    \end{equation}
    The proof for the other statement follows similarly: Consider $\mathbf v=\sum_j c^\prime_j \mathrm S(g_j) \in \mathrm{rowspace}(\mathcal T_\mathcal S)$. Then, $s\equiv \prod_j g_j^{c^\prime_j}\in \mathcal S$, since it is a product of generators.
\end{proof}

We can use the symplectic form of Pauli strings to construct the desired representation of $\psi$ guaranteed by Lemma \ref{lem:fixednullity}. As illustrated by Proposition \ref{prop:stabmembership}, a way to encode the stabilizer group is to specify its generators, given by $\mathcal T_{\mathrm{Stab}(\psi)}$. Since the projector $\Pi$ is specified by its generators:
\begin{equation}
    \Pi = \prod_{j=1}^{\mathrm{rank}(\mathcal S)}\left(\frac{1+ g_j}{2}\right) = \frac{1}{|\mathrm{Stab}(\psi)|}\sum_{P \in \mathrm{Stab}(\psi)}P\;,
\end{equation}
knowledge of the tableau is a valid compression of the projector. Since the bad generators can be represented by $h_i \mapsto (\tr(h_i \psi),\mathrm S(h_i) )$, the final representation is given as:
\begin{equation}
    \psi  = |\psi\rangle \langle \psi|\rightsquigarrow (\mathcal T_{\mathrm{Stab}(\psi)},\{\tr(h_i\psi), \mathrm S(h_i)\}_{i=1}^k )\;,
    \label{eq:fixednullity_repfull}
\end{equation}
shown in simplified form in Eq. (\ref{eq:fixednullity_rep}).

\section{Ensembles with $\mathrm{poly} \log n$ qubits and Proof of Theorem \ref{thm:arealawMDhardness} \label{app:arealawMDhardness}}
We will consider $n = \Lambda_L =\{1,2,\cdots, L\}^D$. The theorem follows by showing that the proof of \cite{gu2024magic} is robust when imposing the area-law. It relies on the properties of two ensembles of states: First, define:
\begin{equation}
    \mathcal E_{\mathrm{Haar}(\nu)} \equiv \{\ket{\phi_\mathrm{Haar}}\otimes \ket{0}^{\otimes (n-\nu)}\;|\; \ket{\phi_\mathrm{Haar}} \in \mathrm{Haar}_\nu\}\;,
    \label{eq:ensembledopedhaar}
\end{equation}
with $\mathrm{Haar}_\nu$ being the Haar ensemble over $\nu$ qubits. Similarly, we also define:
\begin{equation}
    \mathcal E_{g,S,\nu} \equiv \{|\phi_{g,S}\rangle\otimes \ket{0}^{\otimes (n-\nu)}\;|\; \{|\phi_{g,S}\rangle \in \mathcal E_{g,S}(\nu)\}\;,
    \label{eq:ensembledopedsubset}
\end{equation}
where $\mathcal E_{g,S}(\nu)$ is the ensemble of subset phase states defined on $\nu$ qubits, as:
\begin{equation}
    |\phi_{g,S}\rangle = \frac{1}{\sqrt{|S|}}\sum_{\mathbf x \in S }(-1)^{g(\mathbf x)}\ket{\mathbf x}\;,
\end{equation}
with $g:\{0,1\}^\nu\to \{0,1\}$ a random Boolean function and $S \subseteq \{0,1\}^\nu$ a random subset \cite{aaronson2022quantum}. In \cite{gu2024magic}, it was shown the following:
\begin{prop}
    (Theorem 10 of \cite{gu2024magic} and Lemma S6 of \cite{gu2024pseudomagic}) Let $c>  1$. Set $\nu =\Theta(\log^c n)$ and $|S| = \Theta(\exp\log^{c^\prime} n)$, with $c^\prime \in (1,c)$. Two facts hold:
    \begin{enumerate}
        \item $\mathcal E_{\mathrm{Haar}(\nu)}$ and $\mathcal E_{g,S,\nu}$ are MD with high probability;
        \item Consider a state $\psi^\prime  = \ket{\psi^\prime}\bra{\psi^\prime}$ a MD state with $S(\psi^\prime_A) =\omega(\log n)$. If exists a sample efficient protocol that, with high probability: (1) estimate $S(\psi^\prime_A)$ within relative error at most $1/3$ relative error, or (2) a efficient LOCC protocol that distills a constant number of Bell pairs $M_+/S(\psi^\prime_A) = \Omega(1)$ or (3) an efficient LOCC protocol that requires $M_-/S(\psi_A) = O(1)$ to prepare a general MD state even with polynomially many copies of $\psi^\prime$, for all such $\psi^\prime$, then $\mathcal E_{\mathrm{Haar}(\nu)}$ and $\mathcal E_{g,S,\nu}$ are efficiently distinguishable.
    \end{enumerate}
    \label{prop:knownensembleresults}
\end{prop}
The strategy to show Theorem \ref{thm:arealawMDhardness} is to ``localize'' the entanglement of the states in the ensembles, meaning to construct an embedding of the doped states such that they are area-law. To do so, we will need to build some results. 

First, we will need a discrete isoperimetric inequality:
\begin{prop}[\cite{bollobas1991edge}]
    Let $A \subseteq [n]\subseteq \Lambda_D$, with $|A|\leq n/2$. Then, there is a dimension-dependent  positive constant $c_D =O(1)$ such that:
    \begin{equation}
        |\partial A|\geq c_D|A|^{(D-1)/D}\;.
    \end{equation}
    \label{prop:isoperimetric}
\end{prop}
Consider:
\begin{prop}
    Given an embedding $[n]\subseteq \Lambda_D$, $\ell \in \mathbb Z_{\geq 1} $ and $1\leq N_\mathrm{max}\leq n$, there is an algorithm that computes a subset $Q\subseteq [n]$ such that $\mathrm{dist}_{\Lambda_D}(i,j)\geq \ell$ for all $i,j \in Q$ provided that $n = \omega(|Q|\ell^D)$. We will refer to those as $\ell$-sparse embeddings.
    \label{prop:lsparseembedding}
\end{prop}
\begin{proof}
    As an input, we will consider the integer-valued matrix for the distances between the points in $[n]$:
    \begin{equation}
        M = (\mathrm{dist}_{\Lambda_D}(i,j))_{1\leq i,j\leq n} \in \mathbb Z_{\geq 0}^{n\times n}\;.
    \end{equation}
    Denote the ball in $[n]$ as:
    \begin{equation}
        B(i,r;[n])\equiv \{j \in [n]\;:\;\mathrm{dist}_{\Lambda_D}(i,j)\leq r\} = \{j \in [n]\;:\;M_{i,j}\leq r\}\;.
    \end{equation}
    \begin{algorithm}[H]
\caption{Greedy sphere-packing algorithm on $[n]$ \label{alg:greedy}}
    \KwIn{ $(M \in \mathbb Z_{\geq 1 }^{n\times n}, \ell \in \mathbb Z_{\geq 1 }, 1\leq N_\mathrm{max} \leq n)$;}
    \KwOut{$V\subseteq [n]$}
    $S \gets [n]$\\
    $Q \gets \emptyset $\\
    \For{$k \gets 1 \mathrm{\;to\;}N_\mathrm{max}$}{
        \If{$S \neq \emptyset$}{choose $i \in S$  
        \\
        append $i$ to $Q$\\
        $S\gets S\backslash \{j \in [n]\;:\;M_{i,j}\leq \ell-1\}$}
    }
    \Return{$Q$}
\end{algorithm}
    Indeed, the points in the output $V$ are well-separated by $\ell$: Given an $i \in V$, it follows that $\{j\in [n]\;:\;\mathrm{dist}_{\Lambda_D}(i,j)\leq \ell-1\}\cap V = \emptyset$. Hence, for every distinct pair $i,j \in V$, we have that $\mathrm{dist}_{\Lambda_D}(i,j)\geq \ell$. By the iteration procedure above, let $k_\mathrm{max}\leq N_\mathrm{max}$ be the maximum iteration above, defined as:
    \begin{equation}
        k_\mathrm{max}\::\Leftrightarrow \begin{cases}\:[n]\backslash B(i_1,\ell -1;[n])\backslash B(i_2,\ell-1;[n])\backslash \cdots\backslash B(i_{k_\mathrm{max}}, \ell-1;[n]) = \emptyset\;,\\
        =N_\mathrm{max}
        \end{cases}
    \end{equation}
    such that $i_{k} 
    \in [n]\backslash B(i_1,\ell -1;[n])\backslash B(i_2,\ell-1;[n])\backslash \cdots\backslash B(i_{k-1}, \ell-1;[n])$, yielding $Q= \{i_1, i_2, \cdots, i_{k_\mathrm{max}}\}$ as the output in the above algorithm. This means that a sufficient condition to have $k_\mathrm{max} =N_\mathrm{max}$ is to have:
    \begin{equation}
        n >\sum_{k=1}^{N_\mathrm{max}} |B(i_k, \ell-1;[n])|\;.
    \end{equation}
    So, let us upper bound the right-hand size in order to guarantee this condition. We have:
    \begin{equation}
       \sum_{k=1}^{N_\mathrm{max}} |B(i_k, \ell-1;[n])| \leq N_\mathrm{max}|B(0, \ell-1)| \leq N_\mathrm{max}|B(0,\ell)|\;,
    \end{equation}
    where we have bounded the number of points in all the 1-norm balls in $[n]$ by their lattice inclusion. We can compute this number as
    \begin{align}
        |B(0,\ell)| =\left|\left\{(j_1, j_2, \cdots,j_D)\in \mathbb Z^D\;:\;\sum_{1\leq k \leq D}
        |j_k|\leq \ell\right\}\right|=\sum_{k=0}^D {D\choose k}2^k {\ell \choose k}\leq 2^D\left[
        \frac{\ell^D}{D!} + O(\ell^{D-1})\right]\leq c_D \ell^D\;,
    \end{align}
    where, on the second inequality, we counted the number of points by first choosing $k$ out of the $D$ coordinates of $(j_1, j_2, \cdots, j_D)$ to be non-zero, choosing one of the $2^k$ sign configurations, and distributing $\ell$ into $k$ integers. In the corresponding sum, the term with $k=D$ dominates for $\ell \to \infty$, which we assume as $n\to \infty$, and we refer to $c_D$ as the corresponding leading constant, which is only dimension dependent. Of course, $n = \omega(N_\mathrm{max} \ell^D)\Rightarrow n > c N_\mathrm{max}\ell^D$ for any $c=O(1)$ constant, satisfying the constraint given that $|Q| =N_\mathrm{max}$.
\end{proof}

\begin{lem}
    The following two statements are true:
    \begin{enumerate}
        \item There is a $\ell$-sparse embedding $Q\subseteq [n]\subseteq \Lambda_D$ for $D\geq 2$ with $\ell \geq |Q|^{1/(D-1)}$ and $\ell, |Q|=O(\mathrm{poly}\log n)$;
        \item Let $\psi$ be a pure $n$ qubit state with $\ket{\psi} =\ket{\phi_Q}\otimes \ket{0}^{\otimes (n-|Q|)}$, with $\phi_Q$ being a  pure $|Q|$ qubit quantum state supported on $Q$ being the $\ell$-sparse embedding of Part 1. Then, it is area law.
    \end{enumerate}
\label{lem:sparseembeddingarealaw}
\end{lem}
\begin{proof}
By Proposition \ref{prop:lsparseembedding}, the sparse embedding exists as an output of Algorithm \ref{alg:greedy} if $n =\omega(|Q|\ell^D)$. But since we guarantee that $\ell|Q|^D = O(\mathrm{poly}\log n)$, the constraint is automatically satisfied, and hence, $Q\subseteq [n]$ exists.

We now move to show the area law result. Suppose $|A|\leq n/2$, and thus all of its subsets will satisfy Proposition \ref{prop:isoperimetric}.The reduced density matrix is given as $\psi_A  = \mathrm{tr}_{\bar A}\phi_Q\otimes \ket{0}\bra{0}^{\otimes|{\bar Q} \cap A|} $, where $\bar Q$ is the complement of $Q$ in $[n]$. Then, $S(\psi_A) = S(\mathrm{tr}_{\bar A}\phi_Q)$, yielding the upper bound:
    \begin{equation}
        S_\alpha(\psi_A) \leq |Q \cap A|\;.
    \label{eq:entropyintersectionbound}
    \end{equation}
We will upper bound this number using the properties of $Q$. Decompose $A = \bigsqcup_i C_i$, where each $C_i$ is a nearest-neighbor connected component. They are disjoint since their boundaries are disjoint, and we have $|\partial A| =\sum_i |\partial C_i|$.

Now, set $k_i =|Q\cap C_i|$. If $k_i =0$, the particular intersection does not contribute. If $k_i=1$, then $C_i$ is a proper subset of the connected ambient box. Hence, $k_i =1 \leq |\partial C_i|$. For $k_i \geq 2$, set $r = \lfloor (\ell-1)/2\rfloor$. The lattice balls centered at the points of $Q\cap C_i$ are pairwise disjoint: That is, given any distinct pair $q,q^\prime \in Q\cap C_i$, 
\begin{equation}
    B(q, r) \cap B(q^\prime,r) =\emptyset\;,
\end{equation}
since points $\mathrm{dist}_{\Lambda_D}(q,q^\prime)\leq 2r < \ell$. It follows that $|B(q,r)\cap C_i|\geq r+1$ and, therefore, this leads to a lower bound in the number of elements on each of the connected components:
\begin{equation}
    |C_i| \geq \left|\bigcup_{q \in Q\cap C_i}B(q,r) \cap C_i\right| =\sum_{q \in Q \cap C_i } |B(q,r) \cap C_i| \geq \sum_{q \in Q\cap C_i}(r+1) =k_i (r+1)\;.
\end{equation}
Since $r+1 = \lfloor (\ell-1)/2\rfloor + 1 =\lceil \ell /2 \rceil \geq \ell/2$, and therefore, $|C_i|\geq k_i \ell/2$. By the isoperimetric inequality in Proposition \ref{prop:isoperimetric},
\begin{equation}
    |\partial C_i|\geq c_D |C_i|^{(D-1)/D} \geq c_D(k_i \ell/2)^{(D-1)/D} \geq c_D 2^{-(D-1)/D}k_i\;,
\end{equation}
where we have used the fact $\ell \geq |Q|^{1/(D-1)}\geq k_i^{1/(D-1)}$. With that, we estabilish the uniform bound $k_i \leq C_D |C_i|$ for all cases, with $C_D$ a dimension-dependent constant. This shows the area-law, since:
\begin{equation}
    |Q\cap A| = \sum_i |Q\cap C_i| =\sum_i k_i \leq C_D\left(\sum_i |\partial C_i |\right) =C_D|\partial A|\;.
\end{equation}
For $|A|>n/2$, apply the above argument for $\bar A$, that will have $|\bar A|\leq n/2$. Since $\psi$ is pure and each site has at most $2D$ neighbors,
\begin{equation}
    S_\alpha(\psi_A) = S_\alpha(\psi_{\bar A}) \leq C_D |\partial \bar A| \leq 2D C_D |\partial A|\;,
\end{equation}
and thus the area-law holds for any region $A$.
\end{proof}

From that, we note that, if we take $S_\nu \subseteq [n]$ with $|S_\nu|= \nu$ to be the support of the state ensembles $\mathcal E_{\mathrm{Haar}(\nu)}$ and $\mathcal E_{g,S,\nu}$, and impose $S_\nu$ to be a $\ell$-sparse embedding guaranteed by part 1 of Lemma \ref{lem:sparseembeddingarealaw}, note that the states in the ensembles are exactly of the form described in part 2. Hence,
\begin{cor}
    Let $S_\nu \subseteq [n]$ be  a $\ell$-sparse embedding as guaranteed by Lemma \ref{lem:sparseembeddingarealaw} with $|S_\nu|  =\nu$, and let $\mathcal E_{\mathrm{Haar}(\nu)}$ and $\mathcal E_{g,S,\nu}$ be the state ensembles supported on $S_\nu$. Then, they are area-law, for any $S \subseteq [n]$.
\end{cor}

Note that depending on the embedding $S_\nu \subseteq [n]$, the states in Eqs. (\ref{eq:ensembledopedhaar}, \ref{eq:ensembledopedsubset}) may not host an area law even if $\nu = o(n)$. For example, consider the states in \ref{eq:ensembledopedhaar} , where the set of $\nu$ qubits, $S_\nu \subseteq [n]$, forms a $D$-dimensional hypercube in $\Lambda_D$, with $L$ being the linear length such that $\nu =L^D$. Let $A\subseteq [n]$ denote half of the corresponding hypercube. In this case, we have that $|\partial A|  = \Theta(L^{D-1}) = \Theta(\nu^{(D-1)/D})$. If the area law is imposed to $S(\psi_A) = \Theta(\nu)$, we would get a contradiction, since $\Theta(\nu) \neq O(\nu^{(D-1)/D})$ for $\nu = O(\mathrm{poly}\log n)$. Hence, in order to be able for an area law to appear, the doped qubits must lie with sufficient distance between each other, as the packings as described in Proposition \ref{prop:lsparseembedding} describe.

Putting all the pieces together, we have:
\begin{proof}[Proof of Theorem \ref{thm:arealawMDhardness}]
First, by Lemma \ref{lem:sparseembeddingarealaw}, there is an embedding that it an embedding of the qubits such that both ensembles contain only area-law states. Suppose that any of the three tasks exists. Then, by point 2 of Proposition \ref{prop:knownensembleresults}, one can efficiently distinguish the states between $\mathcal E_{\mathrm{Haar}(\nu)}$ and $\mathcal E_{g,S,\nu}$ with high probability, since in this case they are MD and are area-law. 

However, in \cite{aaronson2022quantum}, it was shown that no efficient distinguisher of the two ensembles exists, meaning that $\mathcal E_{g,S,\nu}$ is \emph{pseudorandom}. This follows from the fact that they form approximate designs (with the scaling choices of $\nu$ and $S$ from the Proposition), meaning:
\begin{equation}
    \|\mathbb E_{\psi \sim \mathcal E_{g,S,\nu}} \psi^{\otimes m} -\mathbb E_{\psi \sim \mathcal E_{\mathrm{Haar}(\nu)}} \psi^{\otimes m}\|_1 \leq o\left(\frac{1}{\mathrm{poly}n}\right)\;, \quad \forall m =O(\mathrm{poly} n)\;.
\end{equation}
Thus, the three proposed tasks cannot exist, and (1) efficient protocols to entanglement estimation needs to have $\omega(1)$ error, (2) an efficient LOCC protocol can only distill $M_+/S(\psi_A) = o(1)$ Bell pairs, and (3) dilution requires $M_-/S(\psi_A)  =\omega(1)$ Bell pairs even with polynomially many copies of the target state.
\end{proof}

\section{An argument for volume-law magic of 1d ground states\label{app:vollawmpsargument}}

In Ref. \cite{hastings2007area} (with further improvement in \cite{arad2012improved, arad2013area,huang2014area}), it was shown that any gapped, one-dimensional Hamiltonian with a constant degeneracy has a ground space that can be represented as a matrix product state of sublinear bond dimension. That is, given some basis $\{\psi^1, \cdots, \psi^k\}$ of the ground-state  subspace with $k=O(1)$, their computational-basis components can be compressed as
\begin{align}
    \psi^j_{i_1, i_2, \cdots, i_n} = \includegraphics[width=0.25\textwidth]{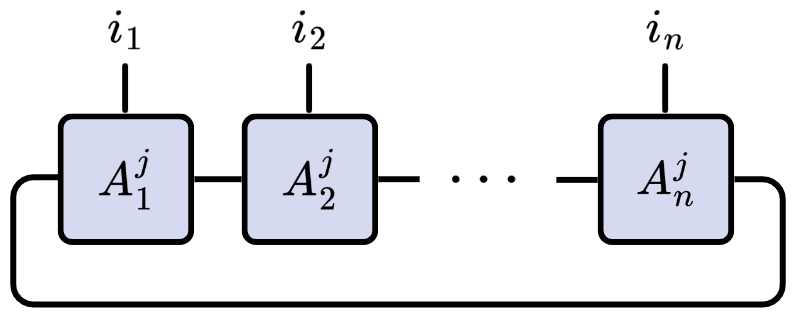}\;,
\end{align}
where the diagram represents $\tr(A^j_{i_1} A^j_{i_2}\cdots A^j_{i_n})$, with $A^j_i$ being $\chi \times \chi$ matrices with $\chi = O(n)$ and we assumed periodic boundary conditions for simplicity. Furthermore, one can always represent a matrix by a unitary embedding, such that
\begin{equation}
    \psi^j_{i_1, i_2, \cdots, i_n}= \includegraphics[width=0.25\textwidth]{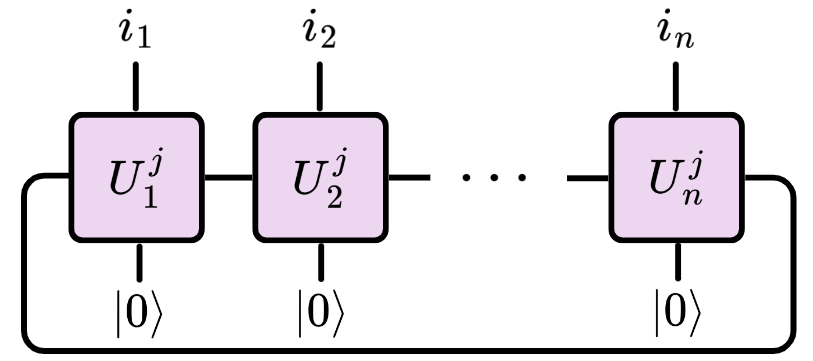}\;,
\end{equation}
where $\{U^j_i: \mathbb C^\chi \times \mathbb C^2 \to \mathbb C^\chi \times \mathbb C^2\}$ are unitaries. In Ref. \cite{chen2024magic}, it was shown that random matrix product states (by constructing a measure $\mu_\mathrm{MPS}$ on the space of states from the Haar measure of the corresponding local unitaries) have extensive stabilizer Rényi entropy with high probability. More specifically, it was shown that
\begin{equation}
    \mathrm{Pr}_{\psi \sim \mu_\mathrm{MPS}}[M_{1/2}(\psi) > \Omega(n)] \geq 1-2^{-\Omega(n)}\;,
\end{equation}
providing a proof that the numerically observed behaviour of Eq. (\ref{eq:magicvolumelaw}) in 1d gapped phases is typical for random matrix product states. That is, the corresponding extensive scaling of the stabilizer Rényi entropy is already seen as the generic behaviour for matrix product states. This argument stems from the ease with which magic can be saturated under local unitaries: For example, another argument comes from the discussion of Ref. \cite{turkeshi2024magic}, where 
it is shown that a one-layer Haar-random quantum circuit $U  = U_1\otimes U_2\otimes \cdots \otimes U_n$ is enough to saturate the stabilizer Rényi entropy of $\ket{\psi} = U \ket{0}^{\otimes n}$ with high probability. 

\section{Proofs of Lemmas}

\subsection{Proof of Lemma \ref{lem:entropygrowth} \label{app:entropygrowthproof}}
It will follow from the following claim:

\noindent \textbf{Claim: }Let $\psi = \ket{\psi}\bra{\psi}$ and $ \psi^\prime = \ket{\psi^\prime}\bra{\psi^\prime}$ be two pure $n$-qubit quantum states in the same $D$-dimensional quantum phase, and let $ A \subseteq [n]$. Then, there is a $c=O(1)$ constant such that:
    \begin{equation}
        S_\alpha(\psi^\prime_A)-S_0(\psi_A) \leq  c|\partial A|\;.
    \end{equation}
Since that implies $S_\alpha(\psi^\prime)\leq S_0(\psi_A) + c |\partial A|$, and given that $\psi$ is area-law by hypothesis, we have that $S_0(\psi_A) = O(|\partial A|)$, yielding an area-law for $\psi^\prime$.

Let us now show the claim: Consider the Schmidt decomposition \cite{nielsen2010quantum} of $\psi$ given by:
    \begin{equation}
        \ket{\psi} =\sum_{i=1}^{\mathrm{rank}(\psi_A)} \lambda_i \ket{i}_A\otimes \ket{i}_{\bar A}\;.
    \end{equation}
Consider the following decomposition of the quantum circuit:
\begin{equation}
    U = \prod_{\ell=1}^{\mathrm{depth}(U)}u_\ell =\prod_{\ell=1}^{\mathrm{depth}(U)} u_{\partial A, \ell}u_{A,\ell}  u_{\bar A, \ell}\;,
\end{equation}
where $\mathrm{depth}(U)=O(1)$, $u_{A,\ell}$ and $u_{\bar A, \ell}$ act exclusively on $A$ and $\bar A$, respectively, while $u_{\partial A, \ell}$ acts on both $A$ and $\bar A$. The strategy to bound the Rényi entropies is the upper bound $S_\alpha(\psi^\prime_A) \leq S_0(\psi^\prime_A)=\log \mathrm{rank}(\psi^\prime_A)$, along with a bound on how big the rank can get. 

Only the unitaries $\{u_{\partial A, \ell}\}_\ell$ can increase the rank of the decomposition, since $u_{A,\ell}$ and $u_{\bar A, \ell}$ corresponds to local unitaries. Unitary operators themselves can be Schmidt-decomposed:
\begin{equation}
    u_{\partial A, \ell} = \sum_{k=1}^{\chi_{\partial A}(\ell)} s_{\partial A,k}(\ell)\; u_{\partial A, k}^A(\ell)\otimes u^{\bar A}_{\partial A,k}(\ell)\;,
    \label{eq:unitaryschmidt}
\end{equation}
where $\chi_{\partial A}(\ell)$ is the Schmidt operator rank of $u_{\partial A,\ell}$, $\{s_{\partial A,k} > 0\}_k$ is a set of real numbers. It is upper bounded by $\chi_{\partial A}(\ell) \leq \min(2^{|\mathrm{supp}_A(\ell)|}, 2^{|\mathrm{supp}_{\bar A}(\ell)|})$, where $\mathrm{supp}_A(\ell)$ is the support of the action of $u_{\partial A,\ell}$ on $A$ and similarly for $\bar A$. They are constrained by geometrical locality. By definition, if $u_\ell = \prod_{j}u_{j,\ell}$, we have, for example for $A$:
\begin{equation}
    u_{\partial A, \ell} = \prod_{j:\;\substack{\mathrm{supp}_A(u_{j,\ell}) \neq \varnothing\\\mathrm{supp}_{\bar A}(u_{j,\ell}) \neq \varnothing} }u_{j,\ell} \quad  \Rightarrow \quad \mathrm{supp}_A(\ell) =  \bigcup_{j:\;\substack{\mathrm{supp}_A(u_{j,\ell}) \neq \varnothing\\\mathrm{supp}_{\bar A}(u_{j,\ell}) \neq \varnothing} } \mathrm{supp}_A(u_{j,\ell})\;,
\end{equation}
where we defined $\mathrm{supp}_A(O) \equiv \mathrm{supp}(O) \cap A$. Consider the maximum distance of this support set to $\bar A$:
\begin{equation}
    \max_{v \in \mathrm{supp}_A(\ell)} \mathrm{dist}_{\Lambda_D}(v, \bar A) = \max_{j, \ell}\{\max_{ v \in \mathrm{supp}(u_{j,\ell})}\mathrm{dist}_{\Lambda_D}(v, \bar A)\}_{j,\ell} \leq  \max_{j, \ell} \mathrm{diam}(\mathrm{supp}(u_{j,\ell})) =O(1)\;,
\end{equation}
where we used the fact that $\mathrm{dist}_{\Lambda_D}(v, \bar A) = \min_{v^\prime \in \bar A}\mathrm{dist}_{\Lambda_D}(v,v^\prime) \leq \mathrm{diam}(\mathrm{supp}(u_{j,\ell}))$, since we know there is a $v^\prime \in \bar A \cap \mathrm{supp}(u_{j,\ell})$. Then, the support lies in a thickening of the boundary set:
\begin{equation}
    \mathrm{supp}_A(\ell) \subseteq \partial A_R \equiv \{v \in A\;:\;\mathrm{dist}_{\Lambda_D}(v, \bar A)\leq R\}\;,
    \label{eq:supportbound}
\end{equation}
with $R=O(1)$. Note also that $\partial A = \partial A_1 \subseteq \partial A_{R\geq 1}$. Furthermore, given any $v^\prime \in \partial A_{R}$, there is a $v \in \partial A $ such that $\mathrm{dist}_{\Lambda_D}(v,v^\prime) \leq R$, due to the inclusion and the fact that the lattice is connected. This shows that:
\begin{equation}
    \partial A_R \subseteq \bigcup_{v \in \partial A} B(v,R)\;,
    \label{eq:boundarybound}
\end{equation}
where $B(v,R) = \{i \in \Lambda_D\;:\; \mathrm{dist}_{\Lambda_D}(i,v) \leq R\}$ is the lattice ball. Combining Eqs. (\ref{eq:supportbound}, \ref{eq:boundarybound}), we derive that:
\begin{equation}
    |\mathrm{supp}_A(\ell)| \leq \left| \bigcup_{v \in \partial A} B(v,R)\right| \leq \sum_{v \in \partial A} |B(v,R)| \leq |B(v,R) |\times |\partial A|\;,
\end{equation}
where we define $c^\prime=  |B(v,R) |$, which is a $O(1)$ $v$-independent constant if $R=O(1)$. Exchanging $A$ and $\bar A$ in the proof above, we derive that both of the supports are upper bounded by $c^\prime |\partial A|$. It follows that:
\begin{equation}
    \chi_{\partial A}(\ell) \leq \min(2^{|\mathrm{supp}_A(\ell)|},2^{|\mathrm{supp}_{\bar A}(\ell)|} ) \leq 2^{c^\prime|\partial A|}\;.
\end{equation}
Now, by Eq. (\ref{eq:unitaryschmidt}), given a quantum state $\ket{\phi}$, $u_{\partial A,\ell} \ket{\phi}$ has Schmidt rank $\chi_{\partial A}(\ell) \times \mathrm{rank}(\phi_A)$. Therefore:
\begin{equation}
    \mathrm{rank}(\psi^\prime_A) =\prod_{\ell=1}^{\mathrm{depth}(U)}\chi_{\partial A}(\ell)\times\mathrm{rank}(\psi_A)\;.
    \label{eq:rankbound}
\end{equation}
Taking the log on both sides in Eq. (\ref{eq:rankbound}) results in:
\begin{equation}
    S_0(\psi_A^\prime) \leq S_0(\psi_A) + \mathrm{depth}(U) c^\prime |\partial A|\;,
\end{equation}
which finalizes the proof, since $c\equiv \mathrm{depth}(U) c^\prime =O(1)$ at finite depth and $S_\alpha(\psi_A^\prime)\leq S_0(\psi_A)$. $\square$

\subsection{Proof of Lemma \ref{lem:stabarealaw} \label{app:stabarealaw}}
Each part can be shown separately:
\begin{enumerate}
    \item By Proposition \ref{prop:reducedstab}, we already know that:
    \begin{equation}
    \psi^\mathcal S_{A} = \frac{1}{2^{|A|}} \sum_{s \in \mathcal S_A}s_A\;.
    \end{equation}
    A straightforward computation confirms that $(\psi^\mathcal S_A)^2 =|\mathcal S_A|/2^{|A|} \psi^\mathcal S_A$, from which we derive that its spectrum is flat, with all non-zero values being $|\mathcal S_A|/2^{|A|} = 2^{-[|A|-\mathrm{rank}(\mathcal S_A)]}$. Note that its rank is given as $\mathrm{rank}(\psi^\mathcal S_A) =2^{|A|-\mathrm{rank}(\mathcal S_A)}$. The corresponding family of Rényi entropies introduced in Eq. (\ref{eq:renyientropy}) can be directly computed as:
    \begin{equation}
        S_\alpha (\psi^\mathcal S_A) = \frac{\log(2^{-\alpha[|A|-\mathrm{rank}(\mathcal{S}_A)]}|\mathrm{rank}(\psi^\mathcal S_A)|)}{1-\alpha} = |A|-\mathrm{rank}(\mathcal S_A)\;.
    \end{equation}
    Recognizing that $\mathrm{Stab}(\psi^\mathcal S_A) =\mathcal S_A$, we arrive at the proof that $S_\alpha (\psi^\mathcal S_A) = \nu(\psi^\mathcal S_A)$ by the definition of the nullity on Eq. (\ref{eq:stabnullity}).
    \item Note that we can split the generating set $S\equiv \{s_1, s_2,\cdots, s_n\}$ as follows:
    \begin{equation}
        S = \underbrace{S_A}_{\mathrm{fully\;supported\;on\;}A} \sqcup\; S_{\partial A}\sqcup \underbrace{S_{\bar A}}_{\mathrm{fully\;supported\;on\;}\bar A}\;.
    \end{equation}
    We then notice that:
    \begin{equation}
        |S| = n=|A|+|\bar A| = \underbrace{|S_A|}_{=\mathrm{rank}(\mathcal S_A)} + |S_{\partial A}| + \underbrace{|S_{\bar A}|}_{\leq |\bar A|} \Rightarrow |S_A| \geq |A| -|S_{\partial A}|\;, 
    \end{equation}
    deriving then that
    \begin{equation}
        \nu(\psi^\mathcal S_A) = |A|-|S_A| \leq |S_{\partial A}|\;,
        \label{eq:nullityareabound}
    \end{equation}
    bounding the nullity by the number of generators with overlapping support on $A$ and $\bar A$. 

    To impose geometrical locality, we can use the same trick employed in App. \ref{app:entropygrowthproof}: First, note that $|S_{\partial A}| \leq |\mathrm{supp}(S_{\partial A})|$, where $\mathrm{supp}(S_{\partial A})=\cup_{i=1}^n \mathrm{supp}(s_i)$. Then, by geometrical locality, there is a $R=O(1)$ such that $\mathrm{supp}(S_{\partial A}) \subseteq \partial A_R$, as Eq. (\ref{eq:supportbound}). From the same reasoning, it follows that:
    \begin{equation}
        |S_{\partial A}| \leq |\mathrm{supp}(S_{\partial A})| \leq |\bigcup_{v \in \partial A} B_{R}(v)| \leq |B_R(v)|\times |\partial A|\;.
    \end{equation}
    By defining $c= |B_R(v)| = O(1)$ and applying Part 1 with Eq. (\ref{eq:nullityareabound}), we show indeed that $S_\alpha(\psi^\mathcal S_A) = \nu(\psi^\mathcal S_A) \leq c |\partial A|$. $\square$
\end{enumerate}

\subsection{Proof of Lemma \ref{lem:stab_renyi_formula} \label{app:stab_renyi_formula}}
The precise formulation of the classical representation is put forward in Lemma \ref{lem:fixednullity}. The proof follows by a somewhat lengthy, but direct computation from the representation. Given $\psi =\ket{\psi}\bra \psi$ with the representation in Eq. (\ref{eq:fixednullity}) and $A \subseteq [n]$, the corresponding reduced state is given as:
\begin{equation}
    \psi = \Pi\cdot K \;\Rightarrow\; \psi_A =\tr_{\bar A} \psi = \Pi_A \cdot K_A\;, 
\end{equation}
where we have defined:
\begin{align}
    K_A &= \frac{|\mathrm{Stab}(\psi_A)|}{2^{|A|}} \sum_{ i \in  I} \langle h_i \rangle_ \psi h^\prime_i\;,\\
    \Pi_A &= \frac{1}{|\mathrm{Stab}(\psi_A)|} \sum_{s_A \in \mathrm{Stab}(\psi_A)}s_A\;,
\end{align}
where $I = \{i \in \{0, \cdots, k\}| \exists P_i \in \mathrm{Stab}(\psi): \tr_{\bar A}(h_i P_i) \neq 0\}$ is the set of indices where the subset of operators in $\{h_i\}_{i \in I}$ survives the partial trace, and $h^\prime_i = h_i P_i$. We compute $\tilde M_2$ by first evaluating the 2-Rényi entropy and the SRE in Eq. (\ref{eq:localstab}). The direct evaluation of 2-Rényi entropy is given by (Lemma 7 of \cite{gu2024magic}):
\begin{align}
    S_2(\psi_A) = \nu(\psi_A)-\log\left(\sum_{i=0}^k \delta_i \langle h_i\rangle_\psi^2\right)\;.
    \label{eq:S2compressed}
\end{align}
We now want to evaluate:
\begin{equation}
    M_2(\psi_A) =- \log\left(\frac{1}{2^{|A|}} \sum_{P_A \in \mathcal P_{|A|}} \langle P_A\rangle_\psi^4\right) \;.
\end{equation}
Consider the Pauli string average:
\begin{align}
    \langle P_A \rangle_\psi = \frac{1}{2^{|A|}} \sum_{i \in I} \sum_{s_A \in \mathrm{Stab}(\psi_A)} \langle h_i\rangle_\psi \tr(P_A h^\prime_i s_A) =  \sum_{i \in I} \sum_{s_A \in \mathrm{Stab}(\psi_A)}\langle h_i \rangle_\psi \delta_{P_A, h^\prime_i s_A}\;,
\end{align}
on which we can compute:
\begin{align}
    M_2(\psi_A) &= -\log \left(\frac{1}{2^{|A|}} \sum_{P_A \in \mathcal P_{|A|}} \sum_{\substack{i_1,i_2, i_3, i_4 \in I\\s^1_A, s^2_A, s^3_A, s^4_A \in \mathrm{Stab}(\psi_A)}} \prod_{\ell=1}^4\langle h_{i_\ell}\rangle_\psi  \delta_{P_A, h^\prime_{i_\ell}s^\ell_A}\right)\\
    &=-\log \left(\frac{1}{2^{|A|}}  \sum_{\substack{i \in I\\s_A \in \mathrm{Stab}(\psi_A)}} \langle h_{i}\rangle_\psi^4\right) = -\log \left(\frac{|\mathrm{Stab}(\psi_A)|}{2^{|A|}}  \sum_{\substack{i \in I}}\langle h_{i}\rangle_\psi^4\right)\\& = \nu(\psi_A) -\log \left(\sum_{\substack{i \in I}}\langle h_{i}\rangle_\psi^4\right) = \nu(\psi_A) -\log \left(\sum_{i=0}^k\delta_i\langle h_{i}\rangle_\psi^4\right) \;, \label{eq:SREcompressed}
\end{align}
where, on the second equality, we have used the fact that $P_A = h^\prime_{i_\ell}s^\ell_A$ for $\ell \in\{1,2,3,4\}$ implies that $h^\prime_{i_\ell} s^\ell_A = h^\prime_{i_{\ell^\prime}} s_A^{\ell^\prime}$ for all $\ell, \ell^\prime$ in the four values. By the decomposition of the Pauli support of $\psi$ into cosets in Eq. (\ref{eq:cosetdecomp}), this implies that all the indices match, meaning $i_{\ell} = i_{\ell^\prime}$. Combining Eqs. (\ref{eq:SREcompressed}, \ref{eq:S2compressed}), we finally obtain the desired formula:
\begin{equation}
    \tilde M_2(\psi_A)= M_2(\psi_A)-S_2(\psi_A) = - \log\left[\frac{\sum_{i=0}^k \delta_i \langle h_i \rangle_\psi^{4}}{ \sum_{i=0}^k \delta_i \langle h_i \rangle_\psi^2}\right]\;.
\end{equation}
We can now introduce the algorithm, given as an input $(\mathcal T_{\mathrm{Stab}(\psi)}, \{\langle h_i\rangle_\psi, \mathbf y_i\}_{i=0}^k)$, where $\mathbf y_i \in \mathbb F_2^{2n}$ is the symplectic representation of $h_i$, as explained in Eq. (\ref{eq:fixednullity_repfull}). We want to check for each $i \in \{0,1, \cdots, k\}$ if $\delta_i=1$, and then add in the sum. Given $A$, we define the projector into $\bar A$ in the symplectic representation as:
\begin{align}
    \mathbf P_{\bar A}: &\mathbb F_2^{2n}\to \mathbb F_2^{2n}\;,\\
    &(\mathbf a, \mathbf b)\mapsto (\mathbf a_{\bar A}, \mathbf b_{\bar A})
\end{align}
where, given $\mathbf z \in \mathbb F_2^{n}$, we define the split $\mathbf z = \mathbf z_A + \mathbf z_{\bar A}$ according to $[n] =A \sqcup \bar A$ (or $\mathbb F_2^{n}\cong \mathbb F_2^{|A|}\oplus \mathbb F_2^{|\bar A|}$). We first claim that:
\begin{equation}
     i \in I \Leftrightarrow \mathbf P_{\bar A}(\mathbf y_i) \in \mathrm{rowspace}[\mathbf P_{\bar A}(\mathcal T_{\mathrm{Stab}(\psi)})]\;.
\end{equation}
To show this, let us denote the rows of $\mathcal T_{\mathrm{Stab}(\psi)}$ as $\{\mathbf x_j = (\mathbf a_j, \mathbf b_j) \in \mathbb F_2^{2n}\}_{j=1}^{\mathrm{rank}[\mathrm{Stab}(\psi)]}$. Then:
\begin{equation}
    \mathbf P_{\bar A}(\mathbf y_i) \in \mathrm{rowspace}[\mathbf P_{\bar A}(\mathcal T_{\mathrm{Stab}(\psi)})] \Leftrightarrow \exists S \subseteq \{1,2,\cdots, \mathrm{rank}[\mathrm{Stab}(\psi)]\}\;:\;\mathbf P_{\bar A}(\mathbf y_i) + \sum_{j \in S} \mathbf P_{\bar A}(\mathbf x_j) =\mathbf 0\;.
\end{equation}
Denoting the corresponding Pauli defined by $S$ as $P_S \in \mathrm{Stab}(\psi)$ by Proposition \ref{prop:stabmembership}, the last condition happens if and only if $\mathrm{supp}(h_i P_S) \subseteq A \Leftrightarrow \tr_{\bar  A}(h_i P_S)\neq 0$, showing the claim. Hence, we reduce our problem to determining if the projections of the symplectic representations $\{\mathbf y_i\}_i$ lie in the corresponding row-space of the Paulis. The full algorithm is as follows:

\begin{algorithm}[H]
\caption{Calculation of $\tilde M_2(\psi_A)$ for a state of nullity $\nu$. \label{alg:M2est}}
    \KwIn{ $(\mathcal T_{\mathrm{Stab}(\psi)}, \{\langle h_i \rangle_\psi, \mathbf y_i\}_{i=0}^k)$;}
    \KwOut{$\tilde M_2(\psi_A)$}
    $r_1 \gets 0$\\
    $r_2 \gets 0$\\
    \For{$i \gets 0 \mathrm{\;to\;}k$}{
        \If{$\mathbf P_{\bar A}(\mathbf y_i) \in \mathrm{rowspace}[\mathbf P_{\bar A}(\mathcal T_{\mathrm{Stab}(\psi)})]$}{$r_1 \gets r_1 +\langle h_i\rangle_\psi^4$  \\
        $r_2 \gets r_2 +\langle h_i\rangle_\psi^2$}
    }
    \Return{$-\log(r_1/r_2)$}
\end{algorithm}
 The runtime of the algorithm is $O(k \times T_\mathrm{row, n \times n})$, where $T_\mathrm{row, n \times n} $ is the runtime of doing a single row-space checking of the projection. Since the Tableau can be at most a $n \times 2n $ rectangular matrix and $\mathbf y_i \in \mathbb F_2^{2n}$, rowspace checking takes $T_\mathrm{row, n \times n} = O(n^3)$ by Gaussian elimination \cite{CLRS4}, and since $k \leq 4^\nu - O(1) $, we have the claimed complexity of $O(4^\nu n^3)$. $\square$

\subsection{Proof of Lemma \ref{lem:localstabbound} \label{app:localstabbound}}

Consider the representation of $\psi$ guaranteed in Lemma \ref{lem:fixednullity}. Then,
\begin{equation}
    \tr(\psi^2)=\frac{1}{2^n} \sum_{P \in \mathrm{supp}(\psi)}\langle P \rangle_\psi^2= \frac{1}{2^n}\sum_{s \in \mathrm{Stab}(\psi)} \sum_{i=0}^k\langle h_i s\rangle_\psi^2 =  \frac{1}{2^{\nu(\psi)}}\sum_{i=0}^k\langle h_i \rangle_\psi^2\;,
\end{equation}
where we used the decomposition fo the support in Eq. (\ref{eq:cosetdecomp}) and the fact that $|\mathrm{Stab}(\psi)|= 2^{n-\nu(\psi)}$. Hence, we have $\tr(\psi^2)=1\Rightarrow \sum_i \langle h_i \rangle_\psi^2=2^{\nu(\psi)}$. We note that, by Lemma \ref{lem:stab_renyi_formula}, $\tilde M_2$ depends on the ratio $\sum_i \delta_i \langle h_i \rangle_\psi^4/\sum_i \delta_i \langle h_i \rangle_\psi^2$, with $\delta_i \in \{0,1\}$. We notice:
\begin{equation}
    1 \underbrace{\leq}_{\langle h_0\rangle_\psi =\delta_0=1} \sum_{i=0}^k \delta_i \langle h_i \rangle_\psi^4 \underbrace{\leq}_{0 \leq \langle h_i \rangle_\psi \leq 1} \sum_{i=0}^k \delta_i \langle h_i \rangle_\psi^2 \underbrace{\leq}_{\sum_i \langle h_i\rangle_\psi^2=2^{\nu(\psi)}} 2^ {\nu(\psi)}\;.
\end{equation}
Therefore,
\begin{equation}
    0 \leq \tilde M_2(\psi_A) = \log \left[\frac{\sum_{i=0}^k \delta_i \langle h_i \rangle_\psi^2}{\sum_{i=0}^k \delta_i \langle h_i \rangle_\psi^4}\right]\leq \nu(\psi)\;.
\end{equation}
$\square$

\subsection{Bounds on the growth of the nullity and proof of Lemma \ref{lem:stabregimerobust} \label{app:stabregimerobust}}
To show the Lemma, we will need a bound on the growth of the stabilizer nullity:
\begin{prop}
    Let $U$ be a $t$-doped quantum circuit of constant gate locality and $\psi^\prime  = U \psi U^\dagger$. Then, there is a constant $c=O(1)$ such that:
    \begin{equation}
        \nu(\psi^\prime) \leq \nu(\psi) + c t\;.
        \label{eq:stabnullitybound}
    \end{equation}
    \label{prop:stabnullitybound}
\end{prop}
\begin{proof}
    We can decompose the quantum circuit as:
    \begin{equation}
        U  = \prod_{j=1}^{t}C_j\cdot u_j\;,
    \end{equation}
    where $\{C_j\}$ are Clifford untaries, and $u_j$ are the non-Clifford gates. The result in Eq. (\ref{eq:stabnullitybound}) can be shown recursively. To do so, consider a non-Clifford gate $u$ and a quantum state $\phi$. Then, note that if $P \in \mathrm{Stab}(\phi)$ and $\mathrm{supp}(P) \cap \mathrm{supp}(u)=\varnothing$, we have:
    \begin{equation}
        |\tr(P u \phi u^\dagger)|\underbrace{=}_{\mathrm{using\;}[P, u^\dagger]= 0 } |\tr(P u^\dagger u \phi)| = |\tr(P \phi)|=1\;,
    \end{equation}
    implying $P \in \mathrm{Stab}(u \phi u^\dagger)$. We want to count the number of Pauli stabilizers that satisfy the disjoint support condition. Consider:
    \begin{equation}
       \mathrm{rank}[\mathrm{Stab}(u \phi u^\dagger)]\geq  |\{P \in \mathrm{gen}[\mathrm{Stab}(\phi)]\;|\; \mathrm{supp}(u) \cap \mathrm{supp}(P) = \varnothing \}| \equiv \mathrm{rank}[\mathrm{Stab}(\phi_{\bar{\mathrm{supp}(u)}})]\;,
    \end{equation}
    since the set of Paulis that stabilize the state in the complement of the support is included in the first set, and we defined the last set by the split of the set of generators as:
    \begin{equation}
        \mathrm{gen}[\mathrm{Stab}(\phi)] = \underbrace{\mathrm{gen}[\mathrm{Stab}(\phi_{\bar{\mathrm{supp}(u)}})]}_{\mathrm{generators\;supported\;at\;}\bar{\mathrm{supp}(u)}}  \sqcup \underbrace{\mathrm{gen}[\mathrm{Stab}(\phi_{\mathrm{supp}(u)})]}_{\mathrm{generators\;supported\;at\;}\mathrm{supp}(u)} \sqcup \; \underbrace{R_u}_{\mathrm{Paulis\;with\;full\;support}}\;,
    \end{equation}
    where $R_u$ is the set of generators with support on both $\mathrm{supp}(u)$ and its complement. Thus we can lower bound as:
    \begin{equation}
        \mathrm{rank}[\mathrm{Stab}(\phi_{\bar{\mathrm{supp}(u)}})] \geq \mathrm{rank}[\mathrm{Stab}(\phi)]-2|\mathrm{supp}(u)|\;,
    \end{equation}
    where we have used the fact  both $\mathrm{rank}[\mathrm{Stab}(\phi_{\mathrm{supp}(u)})]$ and $\mathrm{rank}(R_u)$ are upper bounded by $|\mathrm{supp}(u)|$. Hence,
    \begin{equation}
        \nu(u \phi u^\dagger) = n-\mathrm{rank}[\mathrm{Stab}(u \phi u^\dagger)] \leq n- \mathrm{rank}[\mathrm{Stab}(\phi_{\bar{\mathrm{supp}(u)}})]\leq \nu(\phi) +2 |\mathrm{supp}(u)|\;.
    \end{equation}
    The final formula can be shown by induction: Consider $u = C_j u_j$ and $\phi = \psi_j$. Then, due to the Clifford invariance of the nullity,
    \begin{equation}
        \nu(C_{j} u_{j}  \psi_{j-1} u_j^\dagger C_j^\dagger) = \nu(u_j \psi_{j-1} u_j^\dagger) \leq \nu(\psi_{j-1}) + 2|\mathrm{supp}(u_j)|\;.
    \end{equation} 
    Defining $\psi_1\equiv \psi$, we do indeed arrive at the final claim by repeating the formula $t$ times:
    \begin{equation}
        \nu(\psi^\prime) \leq \nu(\psi) + 2\sum_{j=1}^t|\mathrm{supp}(u_j)|\underbrace{\leq}_{\mathrm{by\;constant\;gate\;locality:}\; \max_{j}|\mathrm{supp}(u_j)| = k =O(1)} \nu(\psi)+ 2 k t\;.
    \end{equation}
\end{proof}
\noindent With this result, the proof of the Lemma is quite direct:
\begin{proof}[Proof of Lemma \ref{lem:stabregimerobust}]
Recall that since $\psi, \psi^\prime$ are in the same stabilizer regime, there is a $t=o(\log n)$-doped circuit relating the two. Note that by assumption, $S(\psi_A^\prime) = \Theta(S(\psi_A))$ for all subregions $A\subseteq [n]$, and thus, $\psi^\prime$ is also area law.

By Proposition \ref{prop:stabnullitybound}, we have:
\begin{equation}
    \frac{\nu(\psi^\prime)}{S(\psi^\prime_A)} \leq \frac{\nu(\psi)}{S(\psi^\prime_A)} + \frac{o(\log n)}{S(\psi^\prime_A)} = \frac{\nu(\psi)}{\Theta (S(\psi_A))} + \frac{o(\log n)}{\Theta(S(\psi_A))}\;, 
    \label{eq:nullityboundafterFDQC}
\end{equation}
where we have used the fact that the entropy scaling is invariant by assuming that they are on the same stabilizer regime.
Taking the $n \to \infty$ limit, the first term vanishes, since $\psi$ is ED and area-law for $A$, and the second term also vanishes, since we are taking $S(\psi_A)  = \omega(\log n)$ by assumption. Hence, $\psi^\prime$ is also ED. 
\end{proof}

\end{document}